\documentclass[10pt,journal,english,twocolumn]{IEEEtran}
\usepackage{blindtext}

\usepackage{babel}
\usepackage[babel]{microtype}
\usepackage[T1]{fontenc}

\usepackage{graphicx}
\usepackage{xcolor}
\usepackage{tikz}
\usepackage{pgfplots}

\usepackage{tabularx}
\usepackage{booktabs}
\usepackage{multirow}
\usepackage{stfloats} %
\usepackage[caption=false,font=footnotesize]{subfig}

\usepackage{amsmath}
\usepackage{amsfonts}
\usepackage{amssymb}
\usepackage{amsthm}
\usepackage{mathtools}
\usepackage{bm}
\usepackage{siunitx}
\usepackage{dsfont}

\usepackage{csquotes}
\usepackage[backend=biber,url=false,style=ieee,isbn=false,doi=true,eprint=true,dashed=false]{biblatex}
\bibliography{literature.bib}
\renewcommand*{\bibfont}{\footnotesize}

\usepackage{hyperref}
\hypersetup{
	pdfauthor={Karl-Ludwig Besser},
	pdftitle={},
	colorlinks=true,
	allcolors=.,
	urlcolor=blue,
}

\usepackage[acronyms,nomain,xindy,nowarn]{glossaries}
\makeglossaries
\newacronym{5g}{5G}{Fifth-Generation Wireless Systems}
\newacronym{ask}{ASK}{amplitude-shift keying}
\newacronym[plural={AWGN}, \glsshortpluralkey={AWGN}]{awgn}{AWGN}{additive white Gaussian noise}

\newacronym{ber}{BER}{bit error rate}
\newacronym{bler}{BLER}{block error rate}
\newacronym{bpsk}{BPSK}{binary phase-shift keying}

\newacronym{cdf}{CDF}{cumulative distribution function}
\newacronym{cnn}{CNN}{convolutional neural network}
\newacronym{csi}{CSI}{channel state information}
\newacronym{csit}{CSI\babelhyphen{nobreak}T}{channel-state information at the transmitter}
\newacronym{csir}{CSI\babelhyphen{nobreak}R}{channel-state information at the receiver}

\newacronym{ecc}{ECC}{error-correcting code}
\newacronym{ecdf}{ECDF}{empirical distribution function}
\newacronym{ee}{EE}{energy efficiency}

\newacronym{ff}{FF}{feed-forward}
\newacronym{ffnn}{FF-NN}{feed-forward neural network}
\newacronym{fsk}{FSK}{frequency-shift keying}

\newacronym{gm}{GM}{Gaussian mixture}

\newacronym{iid}{i.i.d.}{independent and identically distributed}
\newacronym{il}{IL}{information leakage}
\newacronym{iot}{IoT}{internet of things}

\newacronym{lb}{LB}{lower bound}
\newacronym{los}{LoS}{line-of-sight}

\newacronym{mac}{MAC}{multiple access channel}
\newacronym{map}{MAP}{maximum a posteriori}
\newacronym{mc}{MC}{Monte Carlo}
\newacronym{mimo}{MIMO}{multiple-input multiple-output}
\newacronym{ml}{ML}{machine learning}
\newacronym{mop}{MOP}{multi-objective programming problem}
\newacronym{mrc}{MRC}{maximum ratio combining}
\newacronym{msc}{MSC}{modulation and coding scheme}
\newacronym{mse}{MSE}{mean squared error}

\newacronym{nlos}{NLoS}{non-line-of-sight}
\newacronym{nn}{NN}{neural network}
\newacronym{nve}{NVE}{normalized validation error}

\newacronym{pdf}{PDF}{probability density function}
\newacronym{pwc}{PWC}{polar wiretap code}

\newacronym{qam}{QAM}{quadrature amplitude modulation}

\newacronym{ra}{RA}{rearrangement algorithm}
\newacronym{rbf}{RBF}{radial basis function}
\newacronym{relu}{ReLU}{rectified linear unit}
\newacronym{ris}{RIS}{reconfigurable intelligent surface}
\newacronym{rnn}{RNN}{recurrent neural network}

\newacronym{sgd}{SGD}{stochastic gradient descent}
\newacronym{sgp}{SGP}{secure goodput}
\newacronym{simo}{SIMO}{single-input multiple-output}
\newacronym{sinr}{SINR}{signal-to-interference-plus-noise ratio}
\newacronym{siso}{SISO}{single-input single-output}
\newacronym{snr}{SNR}{signal-to-noise ratio}
\newacronym{svm}{SVM}{support vector machine}

\newacronym{tvd}{TVD}{total variation distance}

\newacronym{ub}{UB}{upper bound}
\newacronym{urllc}{URLLC}{ultra-reliable low latency communication}

\newacronym{zoc}{ZOC}{zero-outage capacity}
\setacronymstyle{long-short}
\glsdisablehyper

\newcommand{\diff}{\ensuremath{\mathrm{d}}}
\newcommand{\diag}[1]{\ensuremath{\operatorname{diag}\left(#1\right)}}

\newcommand{\mat}[1]{\boldsymbol{#1}}
\newcommand{\expect}[2][]{\ensuremath{\mathbb{E}_{#1}\left[#2\right]}}
\newcommand{\inv}[1]{\ensuremath{#1^{-1}}}

\DeclarePairedDelimiter{\abs}{\vert}{\vert}

\DeclarePairedDelimiter{\floor}{\lfloor}{\rfloor}

\newcommand{\transpose}[1]{\ensuremath{#1^{\mathsf{T}}}}
\DeclareMathOperator{\var}{var}
\renewcommand{\vec}{\bm}

\DeclareMathOperator{\Ei}{Ei}
\DeclareMathOperator{\E}{E}
\DeclareMathOperator{\step}{\ensuremath{\mathds{1}}}
\newcommand*{\hankel}{\ensuremath{\mathcal{H}}}

\newcommand{\reals}{\ensuremath{\mathds{R}}}
\newcommand{\complexes}{\ensuremath{\mathds{C}}}
\newcommand{\imag}{\ensuremath{\mathrm{j}}}
\newcommand{\e}{\ensuremath{\mathrm{e}}}

\newcommand*{\Q}{\ensuremath{\mathcal{Q}}}
\newcommand*{\unif}{\ensuremath{\mathcal{U}}}
\newcommand*{\normaldist}{\ensuremath{\mathcal{N}}}
\newcommand*{\bernoulli}{\ensuremath{\mathrm{Ber}}}
\newcommand*{\binomialdist}{\ensuremath{\mathcal{B}}}
\newcommand*{\generalbinom}{\ensuremath{\mathcal{GB}}}
\newcommand*{\rayleigh}{\ensuremath{\mathrm{Rayleigh}}}
\newcommand*{\expdist}{\ensuremath{\mathrm{Exp}}}

\newcommand{\tvp}{\ensuremath{\tilde{\varphi}}}

\newcommand{\hLOS}{\ensuremath{h_{\text{LOS}}}}
\newcommand{\pLOS}{\ensuremath{\varphi_{\text{LOS}}}}

\newcommand{\ergCap}{\ensuremath{C_{\text{erg}}}}
\newcommand{\instCap}{\ensuremath{C_{\text{inst}}}}
\newcommand{\epsrate}{\ensuremath{R^{\varepsilon}}}

\newcolumntype{M}{>{\centering\arraybackslash}X}

\theoremstyle{plain}%
\newtheorem{thm}{Theorem}
\newtheorem{cor}{Corollary}
\newtheorem{lem}{Lemma}

\theoremstyle{definition}
\newtheorem{defn}{Definition}%

\theoremstyle{remark}
\newtheorem{rem}{Remark}
\newtheorem*{rem*}{Remark}

\newtheorem*{example*}{Example}

\pgfplotsset{compat=newest}
\pgfplotsset{plot coordinates/math parser=false}
\usetikzlibrary{patterns,arrows,positioning,external,backgrounds,calc,fit,shapes,spy}
\usepgfplotslibrary{fillbetween}

\tikzset{%
	block/.style = {draw, rectangle, minimum height=3em, minimum width=3em},
	riselement/.style = {draw, rectangle, minimum height=.5em, minimum width=.5em},
	ris/.pic = {%
		\foreach \x in {0,...,6}
			\foreach \y in {0,...,4}
				{\node[riselement,anchor=south west] (\x\y) at (.5*\x,0.5*\y) {};}
		\node[above=.1 of 34] {\tikzpictext};
		\node[fit=(00)(04)(64)(60)](ris){};
	}
}
\def\antenna{%
	-- +(0mm,4.0mm) -- +(2.625mm,7.5mm) -- +(-2.625mm,7.5mm) -- +(0mm,4.0mm)
}

\newcommand*\spyfactor{5.1622776601683793319988935444327}
\newcommand*\spypoint{axis cs:6,2.36036}
\newcommand*\spyviewer{axis cs:10,4.2}

\newcommand*\spypointMax{axis cs:48,4.8861}
\newcommand*\spyviewerMax{axis cs:42,3.5}

\newcommand*\spypointErg{axis cs:10,2.915}
\newcommand*\spyviewerErg{axis cs:13,2}

\definecolor{plot0}{HTML}{2db7d2}
\definecolor{plot1}{HTML}{73408C}%
\definecolor{plot2}{HTML}{EB811B}
\definecolor{plot3}{HTML}{14B03D}

\definecolor{change}{HTML}{0096b8}

\title{Reconfigurable Intelligent Surface Phase Hopping for Ultra-Reliable Communications}
\author{Karl-Ludwig Besser, \IEEEmembership{Graduate Student Member,~IEEE} and\\Eduard A. Jorswieck, \IEEEmembership{Fellow,~IEEE}
\thanks{Parts of this work are presented at the 2021 IEEE 22nd International Workshop on Signal Processing Advances in Wireless Communications (SPAWC)~\cite{Jorswieck2021spawc} {and 25th International ITG Workshop on Smart Antennas (WSA)~\cite{Besser2021wsa}}.}
\thanks{The authors are with the Institute of Communications Technology, Technische Universit\"at Braunschweig, 38106 Braunschweig, Germany (email: \{{k.besser}, {e.jorswieck}\}@tu-bs.de).}
\thanks{The work of K.-L. Besser is supported in part by the German Research Foundation (DFG) under grant JO\,801/23-1.}
}

\begin{document}
\maketitle

\begin{abstract}
	We introduce a phase hopping scheme for reconfigurable intelligent surfaces (RISs) in which the phases of the individual RIS elements are randomly varied with each transmitted symbol.
	This effectively converts slow fading into fast fading.
	We show how this can be leveraged to significantly improve the outage performance {especially for small outage probabilities} \emph{without} channel state information (CSI) at the transmitter and RIS.
	Furthermore, the same result can be accomplished even if only two possible phase values are available.
	Since we do not require perfect CSI at the transmitter or RIS, the proposed scheme has no additional communication overhead for adjusting the phases.
	This enables robust ultra-reliable communications with a reduced effort for channel estimation.
\end{abstract}
\glsresetall

\section{Introduction}\label{sec:introduction}
\Glspl{ris} have been considered widely as a promising enabling technology for the next generation of wireless communications to provide a higher throughput, a lower latency,  a better reliability, and an improved security~\cite{di2020smart,Bjornson2020power,ElMossallamy2020}.
They can be used to shape the propagation of electromagnetic waves~\cite{Kaina2014shaping,DelHougne2019}, which can in turn be used to improve wireless data transmission.
A possible use case is the compensation of Doppler effects, e.g., in high-mobility scenarios~\cite{Basar2021dopplerRIS,Huang2021icc}.
Further applications could include localization and sensing~\cite{Bjornson2021,Wymeersch2020}. %
{Due to the special propagation effects in THz channels, \gls{ris}-assisted communications has also been considered for these frequency bands~\cite{Huang2021}.}

The main focus in the aforementioned cases lies on the correct adjustment of the phase shifts of the individual \gls{ris} elements.
Based on the direct channel as well as the \gls{ris}-assisted channel, \cite{Qingqing2019} minimizes the total transmit power at the transmitter by jointly optimizing the transmit beamforming by an active antenna array at the transmitter and reflect beamforming by passive phase shifters at the \gls{ris}.
A similar problem is considered in \cite{Abeywickrama2020}.
An optimization algorithm for finding the optimal \gls{ris} phases that maximize the energy efficiency in a multi-user downlink communication scenario is presented in \cite{Huang2019}.
The maximization of the weighted \gls{sinr} in a two-user downlink network is considered in \cite{Liu2021}.
The optimal phase shifts for a maximum transmission rate of a single-antenna \gls{ris}-assisted communication system are derived in \cite{Bjornson2020}.
{%
In the context of \gls{urllc}, \gls{ris} have been considered in \cite{Ren2021}, where the authors consider a \gls{ris}-assisted factory-automation communication system and investigate the average data rate and decoding error probability.
}

ScatterMIMO exploits smart surfaces to increase the scattering in the environment in order to provide \gls{mimo} spatial multiplexing gain and additional spatial diversity~\cite{Dunna2020}.
By a clever placement of the \gls{ris}, another virtual access point is created whose signals superimpose at the receiver.

However, the various optimization problems to find the optimal phase shifts typically require \gls{csi} at the \gls{ris} or transmitter.
In contrast to this, we do not choose the \gls{ris} coefficients based on \gls{csi} and, thus, do not require \gls{csi} at the transmitter or at the \gls{ris}.
Instead, we propose to use the \gls{ris} to \emph{transform a slow-fading into a fast-fading channel}, in order to improve the reliability of the link.
This is done by randomly varying the \gls{ris} phases with each transmitted symbol during a constant realization of the slow-fading channels.
{Due to some similarities to the well-known frequency hopping~\cite[Chap.~3]{Torrieri2018}, we call the proposed scheme \emph{\gls{ris} phase hopping} in the following.}
The idea stems from the following observation: depending on the antenna geometries, orientation and location of the transmitter, \gls{ris}, and receiver, we can be lucky and obtain a constructive superposition and achieve high data rate or we can be unlucky to get a destructive superposition and an outage.
For ultra-reliable communications it is better to \emph{sacrifice very high peak data rates} to \emph{gain reliability and compensate poor data rates} by averaging over all possible fading states.

{%
The phase hopping scheme that we use in this work is related to the rotate-and-forward scheme from \cite{Pedarsani2015}.
It has already been applied in a similar way in \cite{Nadeem2021twc,Nadeem2021wcl,Psomas2021}.
In \cite{Nadeem2021twc,Nadeem2021wcl}, the authors consider a broadcast channel where \gls{csi} is assumed only at the base station.
Furthermore, they focus on the average sum-rate and \gls{ee} maximization.
Most closely related to our work is \cite{Psomas2021}, where the authors consider a related communication scenario.
However, they consider \gls{nlos} Rayleigh fading which leads to cumbersome expressions of the derived results.
In contrast, we assume an intermittent channel model, e.g., caused by random blockages in a mmWave communication~\cite{Akdeniz2014mmwave,Lin2018intermittent}.
Additionally, we extend the model by also allowing a \gls{los} component between transmitter and receiver.
Furthermore, while the outage probability is investigated in \cite{Psomas2021}, they additionally focus on other quantities like diversity and energy efficiency.
On the other hand, this work focuses on ultra-reliable communications and we specifically show the reliability gained by employing the \gls{ris} phase hopping scheme.
}

Our proposed transmission scheme is based on two different time-scales, which has been considered in a similar way in previous works~\cite{Huang2021icc,Zhao2021,Zhi2021,Zhi2021ergodic}.
In \cite{Huang2021icc}, a two-stage protocol is proposed to mitigate the Doppler effect in a high-mobility scenario.
The protocol includes a training phase in order to adjust the \gls{ris} phases.
In \cite{Zhao2021}, the optimization of the \gls{ris} phases is split into a long-term optimization problem based on statistical \gls{csi} and a short-term optimization based on the faster varying instantaneous \gls{csi}.
Similar ideas of leveraging long-term statistical \gls{csi} are used in \cite{Zhi2021,Zhi2021ergodic}.
However, these previous works again focus on calculating and setting optimal \gls{ris} phases based on \gls{csi}.
If they need to be set to particular values, there is a communication overhead to first estimate the channels, second compute the optimal phases, and third pass them to the \gls{ris}.
In contrast, this is not necessary in our proposed phase hopping scheme.

In the research of metamaterials several different ideas are proposed on how elements that are capable of phase tuning can be designed~\cite{Turpin2014}.
There also exist prototypes of \gls{ris} elements with both continuous phase tuning~\cite{Zhu2013metasurface,Fara2021} and quantized phases, e.g., with down to only two available phase values~\cite{Kaina2014,Kaina2014shaping}.
The influence of discrete phase shifts on the performance of communication systems has also been investigated in \cite{Wu2020discretePhases,Zhang2020discretePhases,Wang2020oneBitPhase}.
For this reason, we additionally investigate the performance of phase hopping under the assumption that only a finite set of possible \gls{ris} phases is available.

The contributions of this work are summarized as follows. %
\begin{itemize}
	\item We present a {\gls{ris}} phase hopping communication scheme in which the phases of the \gls{ris} elements are randomly varied.
	It is shown that this significantly improves the $\varepsilon$-outage capacity for small $\varepsilon$.
	(Section~\ref{sec:varying-phases})
	\item {Furthermore, we analyze the reliability for the \gls{ris}-assisted communication scenario when the phases of the individual \gls{ris} elements can only be adjusted to values from a finite set.}
	(Section~\ref{sec:quant-phases})
	\item {The presented \gls{ris} phase hopping scheme is compared to other phase adjustment schemes, namely static phases and perfect phase adjustment.
	It is shown that for very small tolerated outage probabilities, i.e., in the context of ultra-reliable communications, phase hopping gets close to the performance of perfect phase adjustment \emph{without} requiring \gls{csi} at the \gls{ris}.
	(Section~\ref{sec:comparison})
	}
\end{itemize}
All of the calculations and simulations are made publicly available in interactive notebooks at~\cite{BesserGitlab}.

\textit{Notation:}
Vectors are written in boldface letters, e.g., $\vec{x}$.
For a random variable $X$, we use $F_{X}$, $f_{X}$, and $\phi_{X}$ for its probability distribution function, density function and characteristic function, respectively.
The expectation is denoted by $\mathbb{E}$ and the probability of an event by $\Pr$.
The uniform distribution on the interval $[a,b]$ is denoted as $\unif[a,b]$.
The normal distribution with mean $\mu$ and variance $\sigma^2$ is written as $\normaldist(\mu, \sigma^2)$.
The binomial distribution with $N$ independent trials and success probability $p$ is denoted as $\binomialdist(N, p)$.
The unit step function is written as $\step(x)$.
The real and complex numbers are denoted by $\reals$ and $\complexes$, respectively.
Logarithms, if not stated otherwise, are assumed to be with respect to the natural base.

\section{System Model and Problem Formulation}
Throughout this work, consider a slow-fading \gls{siso} communication system which is assisted by an \gls{ris} with $N$ elements between the transmitter and receiver.
The received signal $y\in\complexes$ is then given as
\begin{equation}
	y = H x + n\,,
\end{equation}
where $x\in\complexes$ is the transmitted signal, $H\in\complexes$ the overall fading coefficient, and $n\in\complexes$ circularly-symmetric complex \gls{awgn}.
{The transmit power is limited by the average power constraint $P$.}
Since we assume an \gls{ris}-assisted communication, the channel fading $H$ is given by~\cite{Huang2019}
\begin{equation}\label{eq:random-phases-overall-gain}
	H = \hLOS + \transpose{\vec{g}}\mat{\Theta}{\vec{h}}\,,
\end{equation}
where $\hLOS\in\complexes$ is the channel coefficient of the \gls{los} connection and $\vec{h}\in\complexes^{N}$ and $\vec{g}\in\complexes^{N}$ represent the channels from transmitter to \gls{ris} and from \gls{ris} to the receiver, respectively. 
The matrix $\mat{\Theta}\in\complexes^{N\times N}$ is a diagonal matrix with the \gls{ris} phases on the main diagonal, i.e., $\mat{\Theta}=\diag{\exp(\imag \theta_1), \dots{}, \exp(\imag \theta_N)}$. 
An illustration of the setup is given in Fig.~\ref{fig:ris-system-model}.
\begin{figure}
	\centering
	\begin{tikzpicture}
	\node[block] (tx) {Transmitter};
	\draw (tx.north) \antenna;

	\pic[above right=1 and 1 of tx, pic text={RIS ($\vec{\Theta}$)}] {ris};
	
	\node[block, below right=1 and 1 of 60.south east, anchor=north west] (rx) {Receiver};
	\draw (rx.north) \antenna;
	
	\draw[->] (tx.east) -- node[above] {$\hLOS$} (rx.west);
	\draw[->] (tx) -- node[above] {$\vec{h}$} (ris);
	\draw[->] (ris) -- node[above] {$\vec{g}$} (rx);
\end{tikzpicture}
	\caption{System model of the \gls{ris}-assisted communication system.}
	\label{fig:ris-system-model}
\end{figure}
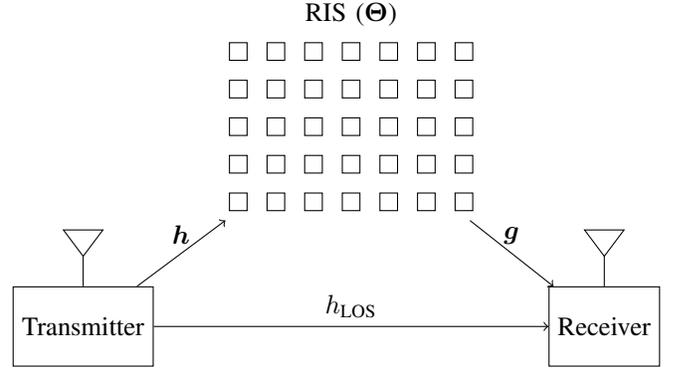

This type of \gls{siso} model can {naturally} originate from an \gls{ris}-aided mmWave \gls{mimo} link (similar to e.g. \cite{Jiguang20}) where the transmit beamforming and receive beamforming are both included in the effective channel vectors $\vec{h}$ and $\vec{g}$, respectively.
{%
Only if the pairing of transmit and receive beamforming vectors is successful, a link is established.
Otherwise, there is an outage, e.g., due to blockage.
This leads to the intermittent fading model considered in the following.
Additionally, only considering a \gls{siso} system provides a lower bound on the performance, which is useful for worst-case design in the context of ultra-reliable communications.}

The \gls{ris} is located in the far-field of the transmitter and receiver.
{Therefore, we assume that both $h_i$ and $g_i$ are independent quasi-static fading coefficients with uniformly distributed phases, i.e., $h_i={\abs{h_i}}\exp(\imag \varphi_i)$ and $g_i={\abs{g_i}}\exp(\imag \psi_i)$ with $\varphi_i, \psi_i\sim\mathcal{U}[0, 2\pi]$.}
{%
For the absolute values $\abs{h_i}$ and $\abs{g_i}$, we assume an intermittent channel model, i.e., $\abs{h_i}$ and $\abs{g_i}$ are independent Bernoulli-distributed random variables~\cite{Lin2018intermittent}.
This model arises from the considered mmWave transmission scenario, since these frequencies are highly susceptible to blockage~\cite{Akdeniz2014mmwave,Alejos2008,Bai2014,Bai2015}.
From this, it is straightforward to see that the \gls{nlos} connection is only available if both $\abs{h_i}$ and $\abs{g_i}$ are equal to \num{1}.
This connection probability is denoted as $p_i = \Pr(\abs{h_i}=1, \abs{g_i}=1)$.
Throughout the following, this will be abbreviated by the random variable $c_i=\abs{h_i}\abs{g_i}$, which is Bernoulli-distributed with $p_i = \Pr(c_i=1)$.
The number of ones in a realization of $\vec{c}=(c_1, \dots{}, c_N)$ is denoted as $\tilde{N}$ and distributed according to a general binomial distribution~\cite{Wang1993}, i.e., $\tilde{N}\sim\generalbinom((p_1, \dots{}, p_N))$.
Note that this corresponds to the regular binomial distribution, if all $p_i$ are the same, i.e., $p_i=p$ for all $i=1, \dots{}, N$.
}
{This model is a slightly modified version of the fluctuating fading model~\cite{Jerez2017}.}
In order to model the {path loss} difference of the \gls{los} component, it has an absolute value of $a$, i.e., we have that
\begin{equation}\label{eq:def-model-los-component}
	\hLOS = a\exp(\imag \pLOS)\,.
\end{equation}
{In a \gls{nlos} scenario, i.e., if there is no \gls{los} connection, we have that $a=0$.}
{Otherwise}, we will typically have $a>1$ in order to reflect a stronger \gls{los} connection.
It is straightforward to normalize the absolute values with respect to $a$, such that the \gls{nlos} components have an absolute value smaller than $1$.
However, for simplifying notation, we will use the above normalization throughout this work.
Additionally, we will assume that $a$ and $\pLOS$ have the same slow fading time-scale as {$\varphi_i$} and $\psi_i$.

Based on these assumptions, we can simplify the expression of $H$ in \eqref{eq:random-phases-overall-gain} to
\begin{align}
	H &= \hLOS + \sum_{i=1}^{N} h_i \left[\mat{\Theta}\right]_{ii} g_i \notag\\
	&= a\exp\left(\imag \pLOS\right) + \sum_{i=1}^{N} {c_i}\exp\left(\imag\left(\varphi_i + \psi_i + \theta_i\right)\right)\,.\label{eq:overall-channel-coeff}
\end{align}
In the case of perfect \gls{csi} at all communication parties, the \gls{ris} phases $\theta_i$ can be optimized based on the channel realizations $\varphi_i$ and $\psi_i$~\cite{Bjornson2020}.
However, while we assume that the receiver has perfect \gls{csi} about the channel realization of $H$, we assume that neither transmitter nor \gls{ris} have \gls{csi}.
This implies that we do not need to estimate the component channels $\hLOS$, $\vec{h}$, and $\vec{g}$, but only the effective channel and only at the receiver side.
At the transmitter, we, therefore, do not perform any power or rate adaption and assume a {constant transmit power}~{$P$} throughout this work.

The suitable performance metrics for this slow fading channel are the outage probability $\varepsilon$ and the $\varepsilon$-outage capacity $\epsrate$~\cite{Tse2005}.
An outage occurs, if the instantaneous channel capacity %
\begin{equation*}
	\instCap = \log_2\left(1 + \abs{H}^2\right)
\end{equation*}
for a (constant) realization of the channels $\hLOS$, $\vec{h}$, and $\vec{g}$ is less than the transmission rate $R$.
The outage probability is, therefore, defined as
\begin{equation}\label{eq:def-out-prob}
	\varepsilon = \Pr\left(\instCap < R\right)\,.
\end{equation}
The $\varepsilon$-outage rate is then defined as the maximum transmission rate for which the probability of an outage is at most $\varepsilon$~\cite{Tse2005},
\begin{equation}\label{eq:def-eps-rate}
	\epsrate = \sup_{R\geq 0}\left\{R \;|\; \Pr\Big(\instCap < R\Big) \leq \varepsilon\right\}\,.
\end{equation}

\subsection{Problem Formulation}
For the communication scenario described above, the following question arises.
\emph{What is a suitable technique to adjust the \gls{ris} phases (without perfect \gls{csi}) in order to achieve a high $\varepsilon$-outage capacity, especially for small $\varepsilon$, e.g., less than $10^{-3}$?}

In this work, we will answer this question by proposing a {phase hopping technique} in which the phases of the individual \gls{ris} elements are randomly changed for each transmitted symbol.
\section{Randomly Varying Phases}\label{sec:varying-phases}
{In order to solve the formulated problem from the previous section, we will apply a \emph{\gls{ris} phase hopping} scheme, which achieves an ultra-low outage probability.}
We summarize the scheme in the following definition.
\begin{defn}[\Gls{ris} Phase Hopping Scheme]
	In the \emph{\gls{ris} phase hopping} scheme, the phases $\theta_i$, $i=1,\dots{},N$, of the $N$ \gls{ris} elements are randomly varied with each transmitted symbol.
	The phase sequence is determined by a pseudo-random sequence which is known at all {legitimate} parties in the communication system.
\end{defn}

\begin{figure}%
	\centering
	{\begin{tikzpicture}%
	\begin{axis}[
		width=.95\linewidth,
		height=.25\textheight,
		xlabel={Transmitted Symbols (Time)},
		ylabel={Phase Values},
		xmin=0,
		xmax=25,
		ymin=0,
		ymax=6.3,
		xtick=\empty,
		xtick={0,8,16,24},
		xticklabels={},
		ymajorgrids,
		yminorgrids,
		xmajorgrids,
		grid style={line width=.1pt, draw=gray!20},
		major grid style={line width=.25pt,draw=gray!40},
		clip=false,
		x label style={at={(axis description cs:0.5,-0.1)},anchor=north},
		]
		\addplot[domain=0:25, samples=50, plot0, thick] {2*pi*rnd};
		\addlegendentry{$\theta_1$};
		\addplot[domain=0:25, samples=50, dashed, plot2, thick] {2*pi*rnd};
		\addlegendentry{$\theta_2$};
		\addplot[very thick, const plot, mark=, black] coordinates {(0, 1.88) (8, 4.97) (16, 2.89) (24,0.3) (25,0.3)};
		\addlegendentry{$\tvp_1$};
		
		\draw[|<->] (axis cs:0,-.3) -- node[fill=white]{\footnotesize Codeword 1} (axis cs:8,-0.3);
		\draw[|<->] (axis cs:8,-.3) -- node[fill=white]{\footnotesize Codeword 2} (axis cs:16,-0.3);
		\draw[|<->|] (axis cs:16,-.3) -- node[fill=white]{\footnotesize Codeword 3} (axis cs:24,-0.3);
	\end{axis}
\end{tikzpicture}}
	\vspace*{-.25em}
	\caption{{Exemplary illustration of the phase hopping scheme. The overall phases $\tvp_i=\varphi_i+\psi_i$ of the slow fading channels $\vec{g}$ and $\vec{h}$ are constant over a long period of time. The phases $\theta_i$ of the individual \gls{ris} elements are randomly varied with each transmitted symbol within one codeword.}}
	\label{fig:phase-hopping-idea}
\end{figure}
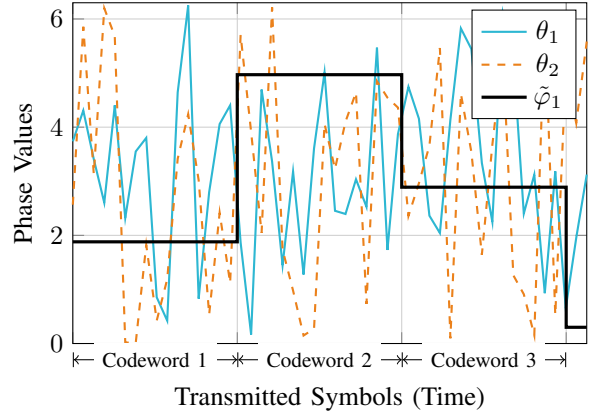

{%
An illustration of the \gls{ris} phase hopping scheme is given in Fig.~\ref{fig:phase-hopping-idea}.
First, it can be seen that for the slow-fading channel, the channel realization $\vec{h}$ and $\vec{g}$ are constant for the transmission of one full codeword.
In contrast,} the phases $\theta_i$ of the \gls{ris} elements are varied randomly at each transmitted symbol.
Since this changing is extremely fast compared to the time-scale of the slow-fading channels, it creates an artificial fast-fading.
However, the phases are adjusted by an underlying pseudo-random sequence that is known at all communication parties, and thus, we obtain perfect \gls{csi} of the fast-fading channel at the receiver.
With this perfect \gls{csir} of the fading realizations (including the artificial fast fading), we can apply the well-known results for the ergodic capacity, which can be achieved by averaging over the received symbols during one constant (slow-fading) realization of the channels $\hLOS$, $\vec{h}$ and $\vec{g}$.
The ergodic capacity is then given by~\cite{Tse2005}
\begin{equation}\label{eq:def-erg-capac}
	\ergCap = \expect[\vec{\theta}]{\log_2\left(1 + \abs*{\hLOS + \sum_{i=1}^{N}{c_i} \exp\left(\imag \left(\theta_i + \tvp_i\right)\right)}^2\right)}\,,
\end{equation}
where we set $\tvp_i = \varphi_i+\psi_i\mod 2\pi$ as the overall phase of the constant channels $\vec{h}$ and $\vec{g}$.

It directly follows, that an outage according to \eqref{eq:def-out-prob} only occurs, if the transmission rate $R$ is larger than the ergodic capacity,
\begin{equation}\label{eq:def-outage-prob-varying}
	\varepsilon = \Pr\left(\ergCap < R\right)\,.
\end{equation}

Throughout this section, we will assume that the \gls{ris} phases~$\theta_i$ are independently and uniformly distributed over $[0, 2\pi]$.
If we had perfect \gls{csi} at the \gls{ris}, the phases $\theta_i$ could be adjusted optimally such that they compensate the phase shifts $\tvp_i$ of the channels, i.e., $\theta_i=-\tvp_i$~\cite{Bjornson2020}.

\begin{rem}[Connection to Frequency-Hopping]\label{rem:frequency-hopping-connection}
	The presented idea of \emph{\gls{ris} phase hopping} {is also referred to as random rotation coding~\cite{Psomas2021}}.
	{However, it} is also closely related to the well-established frequency-hopping method~\cite[Chap.~3]{Torrieri2018}.
	In frequency-hopping, the carrier frequency is frequently changed in order to avoid interference that might occur on specific frequencies.
	Therefore, it helps to also increase the reliability of the transmission.
	Our proposed phase hopping scheme has the following parallels to frequency-hopping.
	In our case, the hop set consists of the possible phase values of the \gls{ris} elements.
	The hop rate is set to match the symbol rate, i.e., we change the phases with each transmitted symbol.
	Just like in frequency-hopping systems, we assume that all users know the hop pattern based on a pseudo-random phase sequence~\cite{Ezzine2020}.
	{Based on these similarities, we propose using the term \emph{\gls{ris} phase hopping}.}
\end{rem}

\subsection{Non-Line-of-Sight Scenario}\label{sub:varying-nlos}
We will start with the simpler \gls{nlos} scenario, i.e., $\hLOS=0$.

{%
\subsubsection{Fixed Number of Available Links}
In order to evaluate the outage probability in \eqref{eq:def-outage-prob-varying}, we need the expression for the ergodic capacity in \eqref{eq:def-erg-capac}.
Since this depends on the number of available links, we first derive the ergodic capacity for a fixed number of available links $\tilde{N}$.
}

\begin{lem}[{{Ergodic Capacity \Gls{nlos} for $\tilde{N}$ Available Links}}]\label{lem:erg-cap-varying-nlos-exact-n-tilde}
	Consider the previously described \gls{ris}-assisted slow fading communication scenario without \gls{csi} at the transmitter and \gls{ris}.
	There is no \gls{los} connection, i.e., $\hLOS=0$.
	The \gls{ris} applies phase hopping with \gls{iid} and uniformly distributed $\theta_i$.
	{Let $\tilde{N}$ out of all $N$ links be available, i.e., $\abs{h_i}=\abs{g_i}=1$ for $i=1, \dots{}, \tilde{N}$ with $\tilde{N}\leq N$.}
	{The ergodic capacity is then given by}
	\begin{equation}\label{eq:erg-cap-varying-nlos-exact}
		C_{\text{erg,NLOS}}{(\tilde{N})} = \int_{0}^{{\tilde{N}}} \log_2\left(1 + s^2\right) \int_{0}^{\infty}s t J_0(st) J_0(t)^{{\tilde{N}}}\diff{t}\,\diff{s}\,.
	\end{equation}
	{%
	The functions $J_0$ and $J_1$ are the Bessel functions of the first kind of orders zero and one, respectively.
	For $\tilde{N}=0$, we have $C_{\text{erg,NLOS}}(0)=0$.
	}
\end{lem}
\begin{proof}
	{First, note that for $\tilde{N}$ available links, we have that
	\begin{equation*}
		S_{\tilde{N}} = \abs*{\sum_{i=1}^{N}\abs{h_i}\abs{g_i}\exp\left(\imag\left(\tvp_i+\theta_i\right)\right)} = \abs*{\sum_{i=1}^{\tilde{N}}\exp\left(\imag\left(\tvp_i+\theta_i\right)\right)},
	\end{equation*}
	since for all non-available links, either $\abs{h_i}$ or $\abs{g_i}$ are zero.
	}
	For \gls{iid} $\theta_i$ with a uniform distribution on $[0, 2\pi]$, the exact \gls{cdf} of $\abs{\sum_{i=1}^{{\tilde{N}}}\exp\left(\imag\left(\tvp_i+\theta_i\right)\right)}$ {can be derived using \cite[Sec.~3.2.1]{Jammalamadaka2001} as %
	\begin{equation}\label{eq:cdf-s-n}
		F_{S_{\tilde{N}}}\left(s\right) = s \int\limits_{0}^{\infty} J_1\left(s\cdot t\right) J_0\left(t\right)^{{\tilde{N}}} \diff{t}\,.
	\end{equation}
	}	
	The corresponding \gls{pdf} $f_{S_N}$ is given by~\cite[Eq.~(3.2.3)]{Jammalamadaka2001}
	\begin{equation}\label{eq:pdf-s-n-exact}
		f_{S_{{\tilde{N}}}}(s) = \int_{0}^{\infty}s t J_0(st) J_0(t)^{{\tilde{N}}}\diff{t},\quad 0\leq s \leq {\tilde{N}}\,.
	\end{equation}
	Combining this with the expectation from \eqref{eq:def-erg-capac} yields the expression in \eqref{eq:erg-cap-varying-nlos-exact}.
\end{proof}

Even though the expression for the ergodic capacity in \eqref{eq:erg-cap-varying-nlos-exact} looks cumbersome, it can be efficiently calculated numerically.
For this, we need the following observation.
The Hankel transform of order $\nu$ of function $f$ is defined as~\cite[Chap.~9]{Poularikas2010}
\begin{equation}\label{eq:def-hankel-transform}
	\hankel_\nu\{f(t)\}(s) = \int_{0}^{\infty} f(t) J_{\nu}(st)t\diff{t}\,.
\end{equation}
With this, we can rewrite \eqref{eq:erg-cap-varying-nlos-exact} in terms of the Hankel transform of $J_0(t)^{\tilde{N}}$ as
\begin{equation}\label{eq:erg-capac-varying-nlos-exact-hankel}
	C_{\text{erg,NLOS}}{(\tilde{N})} = \int_{0}^{{\tilde{N}}} \log_2\left(1 + s^2\right) s \hankel_0\left\{J_0(t)^{{\tilde{N}}}\right\}\!(s) \diff{s}\,.
\end{equation}
This can then be efficiently calculated numerically~\cite{Ogata2005}.
For the results presented in the following, we use the implementation in the \texttt{hankel} library~\cite{Murray2019hankel} in Python.
The source code to reproduce all calculations and simulations can be found in~\cite{BesserGitlab}.

However, this method still requires a specialized implementation of the Hankel transform, which might not be widely available.
We, therefore, present the following approximation for large ${\tilde{N}}$, which is easier to evaluate and useful as a guideline for system design.

\begin{lem}[{Approximate {Ergodic Capacity \Gls{nlos} for $\tilde{N}$ Available Links}}]\label{lem:erg-cap-varying-nlos-n-tilde}
	Consider the previously described \gls{ris}-assisted slow fading communication scenario without \gls{csi} at the transmitter and \gls{ris}.
	There is no \gls{los} connection, i.e., $\hLOS=0$.
	The \gls{ris} applies phase hopping with \gls{iid} and uniformly distributed $\theta_i$.
	{Let $\tilde{N}$ out of all $N$ links be available, i.e., $\abs{h_i}=\abs{g_i}=1$ for $i=1, \dots{}, \tilde{N}$ with $\tilde{N}\leq N$.}
	{%
	The ergodic capacity is then approximated by
	\begin{equation}\label{eq:erg-cap-varying-nlos-approx}
		C_{\text{erg,NLOS}}(\tilde{N}) \approx -\frac{\exp\left(\frac{1}{{\tilde{N}}}\right) \Ei\left(-\frac{1}{{\tilde{N}}}\right)}{\log 2}\,.
	\end{equation}
	For $\tilde{N}=0$, we have $C_{\text{erg,NLOS}}(0)=0$.
	}
\end{lem}
\begin{proof}
For independent $\theta_i$ and with the {normal distribution} approximation for large ${\tilde{N}}$ from \cite[Sec.~3.4.1]{Jammalamadaka2001}, we have that $\abs{\sum_{i=1}^{{\tilde{N}}} \exp\left(\imag\left(\tilde{\varphi_i} + \theta_i\right)\right)}^2$ is exponentially distributed with mean ${\tilde{N}}$.
For this case, the ergodic capacity from \eqref{eq:def-erg-capac} is calculated as
\begin{align}
	C_{\text{erg,NLOS}}{(\tilde{N})} &= \expect[S^2\sim\expdist(1/{\tilde{N}})]{\log_2\left(1+S^2\right)}\\
	&= -\frac{\exp\left(\frac{1}{{\tilde{N}}}\right) \Ei\left(-\frac{1}{{\tilde{N}}}\right)}{\log 2}\,,%
\end{align}
with $\Ei$ being the exponential integral~\cite[Sec.~5.1]{Abramowitz1972}.
\end{proof}

First, we want to verify how accurate the approximation in Lemma~\ref{lem:erg-cap-varying-nlos-n-tilde} is.
In Fig.~\ref{fig:comparison-erg-cap-exact-approx-varying-nlos}, we show the exact ergodic capacity from \eqref{eq:erg-cap-varying-nlos-exact} together with the approximation for large ${\tilde{N}}$ from \eqref{eq:erg-cap-varying-nlos-approx}.
As expected, it can be seen that the approximation becomes more accurate with increasing ${\tilde{N}}$.
At ${\tilde{N}}=6$, the difference between exact and approximate value is around \num{0.035} which corresponds to a relative error of around \SI{1.5}{\percent}.
In contrast, at ${\tilde{N}}=50$, the difference is only around \num{0.0064}, which is a relative error of around \SI{0.13}{\percent}.
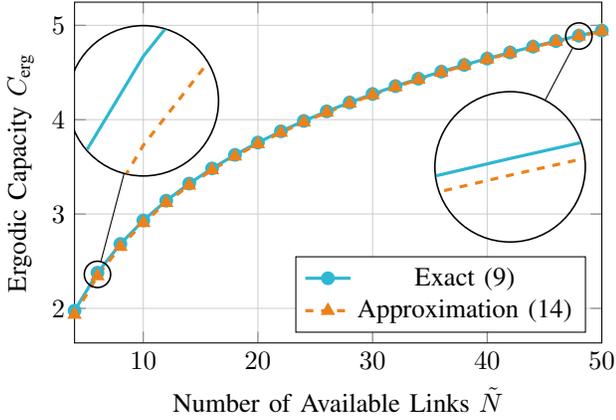
\begin{figure}
	\centering
	\begin{tikzpicture}[spy using outlines={circle, magnification=6, connect spies}]
	\begin{axis}[
		width=.97\linewidth,
		height=.25\textheight,
		xlabel={{Number of Available Links $\tilde{N}$}},
		ylabel={Ergodic Capacity $\ergCap$},
		xmin=4,
		xmax=50,
		legend pos=south east,
		legend entries={{Exact~\eqref{eq:erg-cap-varying-nlos-exact}}, {Approximation~\eqref{eq:erg-cap-varying-nlos-approx}}},
		mark options={solid},
		ymajorgrids,
		xmajorgrids,
		xminorgrids,
		grid style={line width=.1pt, draw=gray!20},
		major grid style={line width=.25pt,draw=gray!40},
		]
		
		\addplot[plot0,mark=*,very thick] table[x=N,y={exact}] {data/approx-exact-random-nlos.dat};
		\addplot[plot2,mark=triangle*,very thick, dashed] table[x=N,y={appr}] {data/approx-exact-random-nlos.dat};

		\node[semithick, circle, draw, minimum size=.35cm, inner sep=0pt] (spypoint) at (\spypoint) {};
		\node[semithick, circle, draw, minimum size=2cm, inner sep=0pt] (spyviewer) at (\spyviewer) {};
		\draw (spypoint) edge (spyviewer);
		\begin{scope}
			\clip (spyviewer) circle (1cm-.5\pgflinewidth);
			\pgfmathparse{\spyfactor^2/(\spyfactor-1)}
			\begin{scope}[scale around={\spyfactor:($(\spyviewer)!\spyfactor^2/(\spyfactor^2-1)!(\spypoint)$)}]
				\addplot[plot0,mark=,very thick] table[x=N,y={exact}] {data/approx-exact-random-nlos.dat};
				\addplot[plot2,mark=,very thick, dashed] table[x=N,y={appr}] {data/approx-exact-random-nlos.dat};
			\end{scope}
		\end{scope}
		
		\node[semithick, circle, draw, minimum size=.35cm, inner sep=0pt] (spypoint2) at (\spypointMax) {};
		\node[semithick, circle, draw, minimum size=2cm, inner sep=0pt] (spyviewer2) at (\spyviewerMax) {};
		\draw (spypoint2) edge (spyviewer2);
		\begin{scope}
			\clip (spyviewer2) circle (1cm-.5\pgflinewidth);
			\pgfmathparse{\spyfactor^2/(\spyfactor-1)}
			\begin{scope}[scale around={\spyfactor:($(\spyviewerMax)!\spyfactor^2/(\spyfactor^2-1)!(\spypointMax)$)}]
				\addplot[plot0,mark=,very thick] table[x=N,y={exact}] {data/approx-exact-random-nlos.dat};
				\addplot[plot2,mark=,very thick, dashed] table[x=N,y={appr}] {data/approx-exact-random-nlos.dat};
			\end{scope}
		\end{scope}
	\end{axis}
\end{tikzpicture}
	\caption{Comparison of the exact and approximate ergodic capacities from \eqref{eq:erg-cap-varying-nlos-exact} and \eqref{eq:erg-cap-varying-nlos-approx}, respectively, for an \gls{nlos} communication system with phase hopping and ${\tilde{N}}$ available links.}
	\label{fig:comparison-erg-cap-exact-approx-varying-nlos}
\end{figure}
Overall, it can be observed in that the approximation from \eqref{eq:erg-cap-varying-nlos-approx} is less than the exact value, i.e., it is a lower bound\footnote{We conjecture that this observation holds true in general for all ${\tilde{N}}$. Unfortunately, we were not able to prove this rigorously at this point.}.
It can, therefore, serve as a worst case design guideline, which is particularly useful for ultra-reliable communication systems.

{%
Both expressions of the ergodic capacity \eqref{eq:erg-cap-varying-nlos-exact} and \eqref{eq:erg-cap-varying-nlos-approx} in Lemma~\ref{lem:erg-cap-varying-nlos-exact-n-tilde} and \ref{lem:erg-cap-varying-nlos-n-tilde}, respectively, are independent of the realization of $\tilde{\varphi_i}$, since it only provides a constant offset for $\theta_i$.
The outage probability from \eqref{eq:def-outage-prob-varying} will therefore be \num{0}, if a rate $R$ less than the ergodic capacity $C_{\text{erg,NLOS}}$ is used for transmission and \num{1} otherwise.
This observation is summarized in the following corollary.
\begin{cor}[{Outage Probability \Gls{nlos} for $\tilde{N}$ Available Links}]\label{cor:out-prob-varying-nlos}
	Consider the previously described \gls{ris}-assisted slow fading communication scenario without \gls{csi} at the transmitter and \gls{ris}.
	There is no \gls{los} connection, i.e., $\hLOS=0$.
	The \gls{ris} applies phase hopping with \gls{iid} and uniformly distributed $\theta_i$.
	Let $\tilde{N}$ out of all $N$ links be available, i.e., $\abs{h_i}=\abs{g_i}=1$ for $i=1, \dots{}, \tilde{N}$ with $\tilde{N}\leq N$.
	The outage probability is then given as
	\begin{equation}
		\varepsilon_{\text{NLOS},\tilde{N}} = \step\left(R - C_{\text{erg,NLOS}}(\tilde{N})\right)\,.
	\end{equation}
\end{cor}
}

{%
\subsubsection{Overall Outage Probability}
For the outage probability in \eqref{eq:def-outage-prob-varying}, we now need to incorporate the individual probabilities that $\tilde{N}$ links are available.
This results in the following theorem for the overall outage probability of an \gls{ris} phase hopping system with $N$ \gls{ris} elements.

\begin{thm}[{Outage Probability \Gls{nlos} with Phase Hopping}]\label{thm:outage-prob-varying}
	Consider the previously described \gls{ris}-assisted slow fading communication scenario without \gls{csi} at the transmitter and \gls{ris}.
	There is no \gls{los} connection, i.e., $\hLOS=0$.
	The \gls{ris} applies phase hopping with \gls{iid} and uniformly distributed $\theta_i$.
	The connection probabilities for all links are the same, i.e., $p_i=p$ for all $i=1, \dots{}, N$.
	Then, the outage probability is given by
	\begin{equation}\label{eq:outage-prob-varying-nlos-overall}
		\varepsilon_{\text{NLOS}}
		= \sum_{i=0}^{N} \step\left(R-C_{\text{erg,NLOS}}(i)\right)\binom{N}{i}p^{i}(1-p)^{N-i}%
		\,,
	\end{equation}
	where $C_{\text{erg,NLOS}}(\tilde{N})$ is the ergodic capacity from \eqref{eq:erg-cap-varying-nlos-exact} or \eqref{eq:erg-cap-varying-nlos-approx}.
\end{thm}
\begin{proof}
	The outage probability is given according to \eqref{eq:def-outage-prob-varying}, which can be written as
	\begin{align*}
		\varepsilon_{\text{NLOS}}
		&= \sum_{i=0}^{N} \Pr\left(C_{\text{erg,NLOS}}(i)<R \;\middle\vert\; \tilde{N}=i\right) \Pr\left(\tilde{N}=i\right)
		\,,
	\end{align*}
	where $\tilde{N}$ describes the number of available links as before, which corresponds to the number of ones in a realization of $\vec{c}=(c_1, \dots{}, c_N)$.
	Since $c_i$ are Bernoulli-distributed with $c_i\sim\bernoulli(p)$, $\tilde{N}$ is binomially distributed with $\tilde{N}\sim\binomialdist(N, p)$.
	Hence, we have that $\Pr(\tilde{N}=i)=\binom{N}{i}p^{i}(1-p)^{N-i}$.
	The outage probability for $\tilde{N}=i$ is given according to Corollary~\ref{cor:out-prob-varying-nlos} as step function with the step occurring at $C_{\text{erg,NLOS}}(i)$.
	Combining this yields \eqref{eq:outage-prob-varying-nlos-overall}.
\end{proof}
}

{%
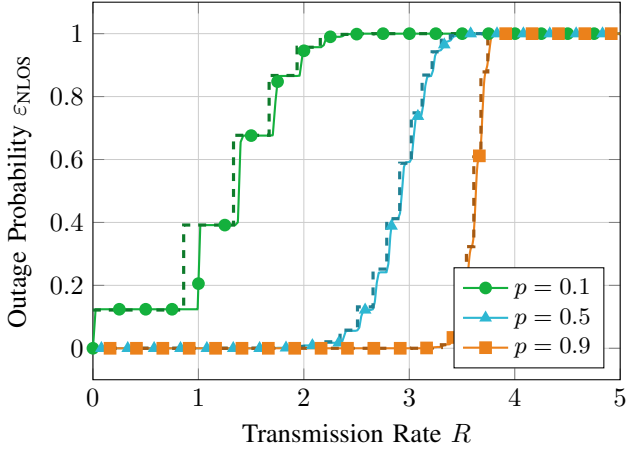
\begin{figure}
	\centering
	\begin{tikzpicture}
	\begin{axis}[
		width=.97\linewidth,
		height=.27\textheight,
		xlabel={Transmission Rate $R$},
		ylabel={Outage Probability $\varepsilon_{\text{NLOS}}$},
		ylabel near ticks,
		xmin=0,
		xmax=5,
		legend pos = south east,
		legend cell align=left,
		legend style={font=\small},
		legend entries = {{$p=0.1$}, {$p=0.5$}, {$p=0.9$}},
		ymajorgrids,
		xmajorgrids,
		xminorgrids,
		grid style={line width=.1pt, draw=gray!20},
		major grid style={line width=.25pt,draw=gray!40},
		]

		\addplot[plot3,mark=*,thick,mark repeat=100] table[x={rate},y={ecdf}] {data/out-prob-random-phase-nlos-N20-p0.100.dat};
		\addplot[plot0,mark=triangle*,thick,mark repeat=100, mark phase=33] table[x={rate},y={ecdf}] {data/out-prob-random-phase-nlos-N20-p0.500.dat};
		\addplot[plot2,mark=square*,thick,mark repeat=100, mark phase=66] table[x={rate},y={ecdf}] {data/out-prob-random-phase-nlos-N20-p0.900.dat};

		\addplot[plot3!70!black,mark=,very thick, dashed,const plot] table[x={rate},y={approx}] {data/out-prob-random-phase-nlos-N20-p0.100.dat};
		\addplot[plot0!70!black,mark=,very thick, dashed,const plot] table[x={rate},y={approx}] {data/out-prob-random-phase-nlos-N20-p0.500.dat};
		\addplot[plot2!70!black,mark=,very thick, dashed,const plot] table[x={rate},y={approx}] {data/out-prob-random-phase-nlos-N20-p0.900.dat};
		
	\end{axis}
\end{tikzpicture}
	\caption{Outage probability for a \gls{nlos} scenario. The phases of the \gls{ris} with {$N=20$} elements are randomly varied. The connection probability of the individual links is $p$. The solid lines show the \gls{ecdf} obtained by Monte Carlo simulations. {For comparison, the dashed lines indicate the approximation for large $N$ from {Lemma~\ref{lem:erg-cap-varying-nlos-n-tilde}}.}}
	\label{fig:out-prob-random-nlos}
\end{figure}
}
The outage probability $\varepsilon_{\text{NLOS}}$ is exemplarily shown in Fig.~\ref{fig:out-prob-random-nlos} for $N=20$ {with different connection probabilities $p$}.
Besides the approximation from {Lemma~\ref{lem:erg-cap-varying-nlos-n-tilde}}, we show results obtained from \gls{mc} simulations with {\num{2000}} slow-fading realizations of $\vec{h}$ and $\vec{g}$, each containing {\num{10000}} fast-fading realizations of $\vec{\theta}$.
The source code to reproduce the figure can be found in \cite{BesserGitlab}.
First, the step-like behavior can immediately be seen from Fig.~\ref{fig:out-prob-random-nlos}.
{%
As expected, the outage probability decreases for an increasing connection probability $p$.
For $p\to 1$, the curve approaches a single step function, i.e., if there are always $N$ connections available, the outage probability is exactly zero for rates below $C_{\text{erg,NLOS}}(N)$ and jumps to one above that threshold.
}

{%
In ultra-reliable communications, the application usually has a tolerated outage probability, typically less than $10^{-5}$.
It is therefore of interest to the communication system designer which rate is the highest, such that the outage probability is at most the tolerated one.
This quantity is known as the $\varepsilon$-outage capacity $\epsrate$~\cite{Tse2005}.
Based on Theorem~\ref{thm:outage-prob-varying}, we specify $\epsrate$ in the following corollary.
\begin{cor}[{$\varepsilon$-Outage Capacity for \Gls{nlos} with \Gls{ris} Phase Hopping}]\label{cor:eps-out-capac-nlos}
	There exist $N+1$ different values of the $\varepsilon$-outage capacities $R^{\varepsilon}$ for $0\leq \varepsilon \leq 1$.
	They are given by the ergodic capacities $C_{\text{erg,NLOS}}(\tilde{N})$ from Lemmas~\ref{lem:erg-cap-varying-nlos-exact-n-tilde} and \ref{lem:erg-cap-varying-nlos-n-tilde} with $\tilde{N}=0, \dots{}, N$.
	In particular, for a given (tolerated) outage probability $\varepsilon$, the $\varepsilon$-outage capacity $\epsrate$ is given as
	\begin{equation*}
		\epsrate = C_{\text{erg,NLOS}}\left(\inv{F_{\tilde{N}}}(\varepsilon)\right)\,,
	\end{equation*}
	where $\inv{F_{\tilde{N}}}$ is the quantile function of $\tilde{N}$.
\end{cor}
Note that this implies that the smallest \emph{non-zero} $\varepsilon$-outage capacity is given by $C_{\text{erg,NLOS}}(1)$.
In this case, exactly \num{1} out of $N$ links is available.
Vice versa, the smallest achievable outage probability is given as $F_{\tilde{N}}(0)$.
For the homogeneous case of equal connection probabilities $p_i=p$, this is simply $(1-p)^{N}$.

This is illustrated in Fig.~\ref{fig:eps-capac-nlos}, where the $\varepsilon$-outage capacities for $N=20$ and $N=50$ are shown for different values of $p$.
As expected, it can be seen that $\epsrate$ increases with an increasing connection probability $p$.
In particular, for $N=20$ and $p=0.1$, the smallest attainable outage probability is around $0.12$, i.e., the $\varepsilon$-outage capacity is zero for tolerated outage probabilities less than around \SI{12}{\percent}.
In contrast, for a higher value of $p=0.5$, the smallest outage probability is around $10^{-6}$.
Similarly, $\epsrate$ increases with a higher number of \gls{ris} elements $N$ for a fixed connection probability $p$.
In the extreme case that $p=1$, i.e., all links are always available, it is even possible to achieve a positive \gls{zoc}~\cite{Besser2021zoc} without perfect \gls{csi} at the transmitter.
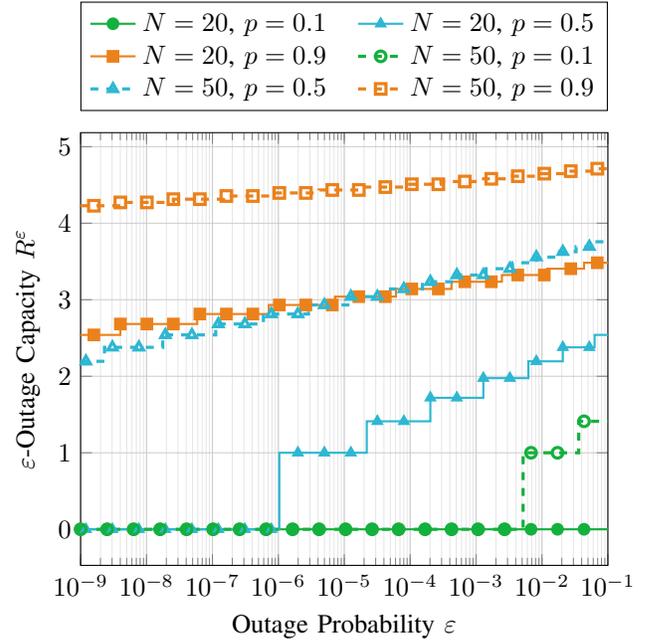
\begin{figure}[t]
	\centering
	\begin{tikzpicture}%
	\begin{axis}[
		width=.97\linewidth,
		height=.3\textheight,
		xlabel={Outage Probability $\varepsilon$},
		ylabel={$\varepsilon$-Outage Capacity $\epsrate$},
		xmin=1e-9,
		xmax=1e-1,
		xmode=log,
		legend pos = south east,
		legend cell align=left,
		legend columns = 2,
		legend style = {
			at = {(1, 1.05)},
			anchor = south east,
			/tikz/every even column/.append style={column sep=0.33cm}
		},
		ymajorgrids,
		xmajorgrids,
		xminorgrids,
		grid style={line width=.1pt, draw=gray!20},
		major grid style={line width=.25pt,draw=gray!40},
		mark options={solid},
		]
		\addplot[plot3,mark=*,thick,mark repeat=10, const plot] table[x=eps,y=exact] {data/eps-capac-N20-p0.100.dat};
		\addlegendentry{$N=20$, $p=0.1$};
		\addplot[plot0,mark=triangle*,thick,mark repeat=10, mark phase=3, const plot] table[x=eps,y=exact] {data/eps-capac-N20-p0.500.dat};
		\addlegendentry{$N=20$, $p=0.5$};
		\addplot[plot2,mark=square*,thick,mark repeat=10, mark phase=6, const plot] table[x=eps,y=exact] {data/eps-capac-N20-p0.900.dat};
		\addlegendentry{$N=20$, $p=0.9$};
		
		\addplot[plot3,mark=o,very thick,mark repeat=10, const plot, dashed] table[x=eps,y=exact] {data/eps-capac-N50-p0.100.dat};
		\addlegendentry{$N=50$, $p=0.1$};
		\addplot[plot0,mark=triangle,very thick,mark repeat=10, mark phase=3, const plot, dashed] table[x=eps,y=exact] {data/eps-capac-N50-p0.500.dat};
		\addlegendentry{$N=50$, $p=0.5$};
		\addplot[plot2,mark=square,very thick,mark repeat=10, mark phase=6, const plot, dashed] table[x=eps,y=exact] {data/eps-capac-N50-p0.900.dat};
		\addlegendentry{$N=50$, $p=0.9$};
	\end{axis}
\end{tikzpicture}
	\vspace*{-.25em}
	\caption{$\varepsilon$-outage capacity for an \gls{nlos} scenario. The phases of the \gls{ris} with $N$ elements are randomly varied. The connection probability of the individual links is $p$.}
	\label{fig:eps-capac-nlos}
\end{figure}
}

{%
For the results above, we assume an intermittent channel model for $\abs{h_i}$ and $\abs{g_i}$.
In the following, we provide an approximation of the outage probability for more general distributions of $\abs{h_i}$ and $\abs{g_i}$.
However, it should be noted that this might not always yield closed-form solutions.

\begin{thm}[{Outage Probability \Gls{nlos} with Phase Hopping and General Fading Distribution}]\label{thm:out-prob-general-fading}
	Consider the previously described \gls{ris}-assisted slow fading communication scenario without \gls{csi} at the transmitter and \gls{ris}.
	There is no \gls{los} connection, i.e., $\hLOS=0$.
	The \gls{ris} applies phase hopping with \gls{iid} and uniformly distributed $\theta_i$.
	The path losses $\abs{h_i}$ and $\abs{g_i}$ are distributed according to $F_{\abs{h_i}}$ and $F_{\abs{g_i}}$, respectively.
	The outage probability is then approximated for large $N$ by
	\begin{equation}
		\varepsilon_{\text{NLOS}} \approx F_{\sigma^2}\left(\frac{1}{2\inv{\E}\left(R \log(2)\right)}\right)\,,
	\end{equation}
	where $F_{\sigma^2}$ is the distribution function of
	\begin{equation}
		\sigma^2 = \frac{1}{2}\sum_{i=1}^{N}\abs{h_i}^2\abs{g_i}^2
	\end{equation}
	and $\inv{\E}$ the inverse function of $\E(x)=-\exp(x)\Ei(-x)$ for $x>0$.
\end{thm}
\begin{proof}
	The proof can be found in Appendix~\ref{app:general-distributions}.
\end{proof}
\begin{rem}
	The distribution function of $\sigma^2$ can be determined by its characteristic function $\phi_{\sigma^2}$, if $\abs{h_i}^2$ and $\abs{g_i}^2$ admit characteristic functions $\phi_{\abs{h_i}^2}$ and $\phi_{\abs{g_i}^2}$, respectively.
	First, one needs to determine the characteristic function $\phi_{c_i^2}$ of $c_i^2=\abs{h_i}^2\abs{g_i}^2$.
	For independent $c_i^2$, it immediately follows that
	\begin{equation}
		\phi_{\sigma^2}(t) = \prod_{i=1}^{N} \phi_{c_i^2}\left(\frac{t}{2}\right) = \left(\phi_{c_1^2}\left(\frac{t}{2}\right)\right)^N
	\end{equation}
\end{rem}

}

{%
\begin{rem}[System Design for Ultra-Reliable Communications]
	Assume that a setup with a \gls{ris} of fixed size $N$ and tolerated outage probability $\varepsilon$ is given.
	A system designer, can use the above results to adjust the transmission rate such that the outage requirement is met.
	First, the blockage/connection probability~$p$ needs to be estimated.
	Second, the $\varepsilon$-outage capacity $\epsrate$ is calculated according to Corollary~\ref{cor:eps-out-capac-nlos}.
	Third, the transmission rate of the system needs to be adjusted to be less than $\epsrate$.
\end{rem}
}
\subsection{Line-of-Sight Scenario}\label{sub:varying-los}
If we have an additional \gls{los} component, the efficient channel is given as
\begin{equation}
	H = {a\exp(\imag \pLOS) + \sum_{i=1}^{N}{c_i}\exp\left(\imag(\tilde{\varphi}_i + \theta_i)\right)}\,,
\end{equation}
where $a>0$ is the fixed absolute value of the \gls{los} component and $\pLOS\sim\unif[0, 2\pi]$ its slow-fading phase.

{%
The outage probability can be approximated similarly to the \gls{nlos} scenario.
}

\begin{thm}[Outage Probability \Gls{los} with Phase Hopping]\label{thm:outage-prob-varying-los}
	Consider the previously described \gls{ris}-assisted slow fading communication scenario without \gls{csi} at the transmitter and \gls{ris}.
	There exists a \gls{los} connection with absolute value $a$ between transmitter and receiver.
	The \gls{ris} applies phase hopping with \gls{iid} and uniformly distributed $\theta_i$.
	{The connection probabilities for all \gls{nlos} links are the same, i.e., $p_i=p$ for all $i=1, \dots{}, N$.}
	The outage probability for this scenario can be approximated as
	{%
	\begin{equation}\label{eq:outage-prob-varying-los}
		\varepsilon_{\text{LOS}} =
		\sum_{i=0}^{N}\step\left(R - C_{\text{erg,LOS}}(i)\right)\binom{N}{i}p^{i}(1-p)^{N-i}\,,
	\end{equation}
	}
	with
	\begin{equation}\label{eq:erg-capac-random-los-integral}
		\begin{split}
			&C_{\text{erg,LOS}}{(\tilde{N})} \approx\\
			&\quad \int_{0}^{\infty}\frac{1}{{\tilde{N}}}\log_2(1+s)\exp\left(\frac{-(a^2+s)}{{\tilde{N}}}\right) I_0\left(\frac{2a}{{\tilde{N}}}\sqrt{s}\right)\diff{s} \,,
		\end{split}
	\end{equation}
	and $I_0$ being the modified Bessel function of the first kind and order zero~\cite[Eq.~9.6.16]{Abramowitz1972}.
	{For $\tilde{N}=0$, we have that $C_{\text{erg,LOS}}(0)=\log_2(1+a^2)$.}
\end{thm}
\begin{proof}
The proof can be found in Appendix~\ref{app:proof-thm-outage-prob-los}.
\end{proof}

The ergodic capacity in \eqref{eq:erg-capac-random-los-integral} has no known closed-form expression, however, it can be efficiently calculated numerically.
This is used for the results that are presented in the following.
The source code to reproduce the calculations and simulations can be found in \cite{BesserGitlab}.

{%
Similarly to the \gls{nlos} case, we use $\tilde{N}$ in \eqref{eq:erg-capac-random-los-integral} to describe the number of available \gls{nlos} links.
}
In Fig.~\ref{fig:ergodic-capac-random-los}, the ergodic capacity $C_{\text{erg,LOS}}$ is shown for different values of the strength of the \gls{los} connection $a$ over the number of {available \gls{nlos} links $\tilde{N}$}. %
The value $a=0$ corresponds to the \gls{nlos} scenario.
The solid lines show the results of the approximation for large {$\tilde{N}$} from \eqref{eq:erg-capac-random-los-integral}.
For comparison, the dashed lines indicate results of \gls{mc} simulations with \num{1000} slow $\times$ \num{5000} fast fading samples.
First, it can be seen that the approximation for large {$\tilde{N}$} matches the simulation results accurately already for ${\tilde{N}}\geq 10$.
As expected, the ergodic capacity increases with increasing $a$ and also with increasing {$\tilde{N}$}.
{For larger $a$, the slope of the curve gets flatter, since the \gls{los} component dominates the overall channel gain.}
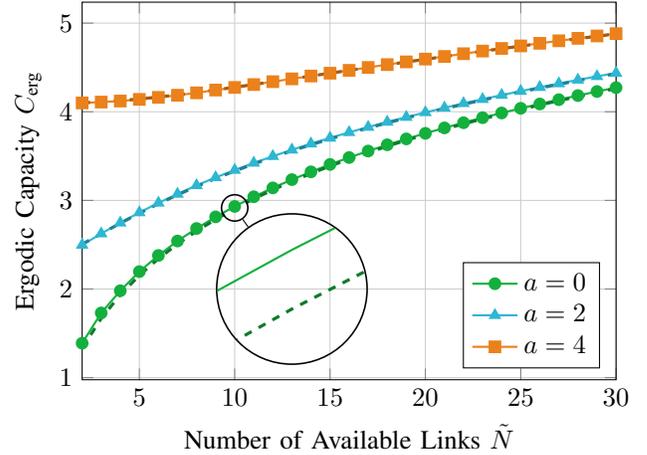
\begin{figure}
	\centering
	\begin{tikzpicture}%
	\begin{axis}[
		width=.98\linewidth,
		height=.27\textheight,
		xlabel={{Number of Available Links $\tilde{N}$}},
		ylabel={Ergodic Capacity $\ergCap$},
		xmin=2,
		xmax=30,
		legend pos=south east,
		legend entries={{$a=0$}, {$a=2$}, {$a=4$}},
		ymajorgrids,
		xmajorgrids,
		xminorgrids,
		grid style={line width=.1pt, draw=gray!20},
		major grid style={line width=.25pt,draw=gray!40},
		]
		
		\addplot[plot3,mark=*,thick] table[x=N,y={mc}] {data/erg-capac-random-phase-nlos-N30.dat};
		\addplot[plot0,mark=triangle*,thick] table[x=N,y={mc}] {data/erg-capac-random-phase-los-a2.00-N30.dat};
		\addplot[plot2,mark=square*,thick] table[x=N,y={mc}] {data/erg-capac-random-phase-los-a4.00-N30.dat};
		
		\addplot[plot3!70!black,mark=,very thick,dashed] table[x=N,y={appr}] {data/erg-capac-random-phase-nlos-N30.dat};
		\addplot[plot0!70!black,mark=,very thick,dashed] table[x=N,y={appr}] {data/erg-capac-random-phase-los-a2.00-N30.dat};
		\addplot[plot2!70!black,mark=,very thick,dashed] table[x=N,y={appr}] {data/erg-capac-random-phase-los-a4.00-N30.dat};

		\node[semithick, circle, draw, minimum size=.35cm, inner sep=0pt] (spypoint) at (\spypointErg) {};
		\node[semithick, circle, draw, minimum size=2cm, inner sep=0pt] (spyviewer) at (\spyviewerErg) {};
		\draw (spypoint) edge (spyviewer);
		\begin{scope}
			\clip (spyviewer) circle (1cm-.5\pgflinewidth);
			\pgfmathparse{\spyfactor^2/(\spyfactor-1)}
			\begin{scope}[scale around={\spyfactor:($(\spyviewerErg)!\spyfactor^2/(\spyfactor^2-1)!(\spypointErg)$)}]
				\addplot[plot3,mark=,thick] table[x=N,y={mc}] {data/erg-capac-random-phase-nlos-N30.dat};
				\addplot[plot3!70!black,mark=,very thick,dashed] table[x=N,y={appr}] {data/erg-capac-random-phase-nlos-N30.dat};
			\end{scope}
		\end{scope}
	\end{axis}
\end{tikzpicture}
	\caption{Ergodic capacity for an \gls{los} scenario. The phases of the \gls{ris} with $N$ elements are randomly varied. The solid lines show the \gls{ecdf} obtained by Monte Carlo simulations. For comparison, the dashed lines indicate the approximation for large {$\tilde{N}$} from \eqref{eq:erg-capac-random-los-integral}.}
	\label{fig:ergodic-capac-random-los}
\end{figure}

Next, we show the outage probability for $a=3$ and $N=20$ for {different values of $p$} in Fig.~\ref{fig:out-prob-random-los}.
The theoretical curves obtained by the approximation for large $N$ from \eqref{eq:outage-prob-varying-los} are indicated by dashed lines.
{%
The first interesting fact, which can be observed, is that for all $p$, we can achieve a positive \gls{zoc}, i.e., transmit at a positive rate without any outages.
This is due to the assumed model with a constantly available \gls{los} connection and not possible in general~\cite{Besser2021zoc}.
For $p=0$, only the \gls{los} connection is available and the ergodic capacity equals $\log_2(1+\abs{a}^2)$.
The outage probability is a step function with step at this point and the \gls{zoc} is equal to this ergodic capacity.
While this \gls{zoc} is the same for all $p$, it can be seen that the $\varepsilon$-outage capacity for $\varepsilon>0$ increases with increasing $p$, i.e., the available \gls{nlos} links help to improve the reliability.
Similar to the \gls{nlos} case, the outage probability function tends to a step function for $p=1$, i.e., when all \gls{nlos} connections are available all the time.
}
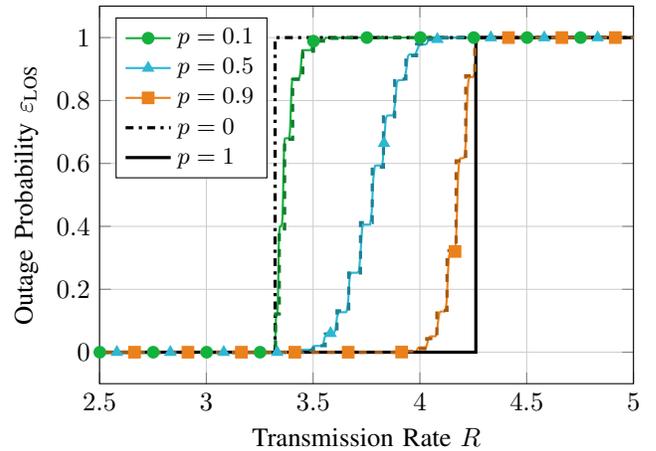
\begin{figure}
	\centering
	{\begin{tikzpicture}
	\begin{axis}[
		width=.98\linewidth,
		height=.27\textheight,
		xlabel={Transmission Rate $R$},
		ylabel={Outage Probability $\varepsilon_{\text{LOS}}$},
		ylabel near ticks,
		xmin=2.5,
		xmax=5,
		legend pos = north west,
		legend cell align=left,
		legend style={font=\small},
		ymajorgrids,
		xmajorgrids,
		xminorgrids,
		grid style={line width=.1pt, draw=gray!20},
		major grid style={line width=.25pt,draw=gray!40},
		]

		\addplot[plot3,mark=*,thick,mark repeat=100] table[x={rate},y={ecdf}] {data/out-prob-random-phase-los-a3.00-N20-p0.100.dat};
		\addlegendentry{$p=0.1$};
		\addplot[plot0,mark=triangle*,thick,mark repeat=100, mark phase=33] table[x={rate},y={ecdf}] {data/out-prob-random-phase-los-a3.00-N20-p0.500.dat};
		\addlegendentry{$p=0.5$};
		\addplot[plot2,mark=square*,thick,mark repeat=100, mark phase=66] table[x={rate},y={ecdf}] {data/out-prob-random-phase-los-a3.00-N20-p0.900.dat};
		\addlegendentry{$p=0.9$};
		
		\addplot[black, very thick, no marks, dashdotted] coordinates {(2, 0) (3.322, 0) (3.322, 1) (5, 1)};
		\addlegendentry{$p=0$};
		\addplot[black, very thick, no marks] coordinates {(2, 0) (4.262, 0) (4.262, 1) (5, 1)};
		\addlegendentry{$p=1$};
		
		\addplot[plot3!70!black,mark=,very thick, dashed,const plot] table[x={rate},y={approx}] {data/out-prob-random-phase-los-a3.00-N20-p0.100.dat};
		\addplot[plot0!70!black,mark=,very thick, dashed,const plot] table[x={rate},y={approx}] {data/out-prob-random-phase-los-a3.00-N20-p0.500.dat};
		\addplot[plot2!70!black,mark=,very thick, dashed,const plot] table[x={rate},y={approx}] {data/out-prob-random-phase-los-a3.00-N20-p0.900.dat};
	\end{axis}
\end{tikzpicture}}
	\caption{Outage probability for an \gls{los} scenario with $a=3$. The phases of the \gls{ris} with $N{=20}$ elements are randomly varied. {The connection probability of all links is $p$.} The solid lines show the \gls{ecdf} obtained by Monte Carlo simulations. For comparison, the dashed lines indicate the approximation for large $N$ from \eqref{eq:outage-prob-varying-los}. {The black curves represent the extreme cases of $p=0$ and $p=1$.}}
	\label{fig:out-prob-random-los}
\end{figure}

\section{Quantized Phases}\label{sec:quant-phases}
In the previous section, we analyzed the outage performance of an \gls{ris}-assisted communication system where the phases of the individual \gls{ris} elements could be set to arbitrary values.
However, in practice this might not be possible and only a discrete set of phase values may be available.

As in Section~\ref{sec:varying-phases}, we assume that the phases $\theta_i$ of the $N$ \gls{ris} elements change with each transmitted symbol.
However, the values of $\theta_i$ are now from a discrete set of phases $\Q$, i.e., we have
\begin{equation}\label{eq:quant-phases-set}
	\theta_i \in \Q=\left\{k \frac{2\pi}{K} \;\middle|\; k=0, \dots{}, K-1\right\}, \quad i=1, \dots{}, N,
\end{equation}
where $K$ is the number of quantization steps.

The expected value in the ergodic capacity expression from \eqref{eq:def-erg-capac} is then simply a weighted sum
\begin{equation}\label{eq:erg-capac-discrete}
	\begin{split}
		&C_{\text{erg}}=\\ &\sum_{\vec{\theta}\in \Q^N} \Pr(\vec{\theta}){\log_2\left(1 + \abs*{\hLOS + \sum_{i=1}^{N}{c_i} \exp\left(\imag\left(\tilde{\varphi_i} + \theta_i\right)\right)}^2\right)}\,.
	\end{split}
\end{equation}
Evaluating this expression exactly can be cumbersome for general $K$ and $N$, since it involves calculating all combinations of phases.
Fortunately, for sufficiently large $N$, we can apply the central limit theorem to obtain approximate results, which will be shown in the following.

\begin{thm}[Outage Probability with Phase Hopping and Quantized Phases]\label{thm:outage-prob-quant}
	Consider the previously described \gls{ris}-assisted slow fading communication scenario without \gls{csi} at the transmitter and \gls{ris}.
	There is a possible \gls{los} connection with absolute value $a$ between transmitter and receiver.
	{The connection probabilities of all \gls{nlos} links are the same, i.e., $p_i=p$ for all $i=1, \dots{}, N$.}
	The phases $\theta_i$ of the $N$ \gls{ris} elements are from the finite set $\Q$ as defined in \eqref{eq:quant-phases-set}.
	The \gls{ris} applies phase hopping with \gls{iid} and uniformly distributed $\theta_i$, i.e., $\theta_i\stackrel{\text{iid}}{\sim}\unif(\Q)$.
	For large $N$, the outage probability for this scenario can be approximated
	according to \eqref{eq:outage-prob-varying-los} for the \gls{los} case.
	In the case of $a=0$, i.e., a \gls{nlos} scenario, the outage probability is approximated by \eqref{eq:outage-prob-varying-nlos-overall}.
\end{thm}
\begin{proof}
{The proof can be found in Appendix~\ref{app:proof-quant-phases}.}
\end{proof}

The important implication of Theorem~\ref{thm:outage-prob-quant} is that the outage probability for a phase hopping system with phase quantization is asymptotically equal to the case with continuous phases, independent of the number of quantization levels.
For \gls{ris} with a large number of elements, we can therefore apply the results from the previous section as a design guideline, even if the \gls{ris} phases can only be set to a finite number of values.

In Fig.~\ref{fig:out-prob-quant-phases-n20}, we show the outage probability $\varepsilon$ for a \gls{nlos} scenario and quantized phases for a \gls{ris} {with $N=20$} elements {and $p=0.5$}.
The number of quantization levels is {varied}.
The solid lines indicate the \glspl{ecdf} obtained by \gls{mc} simulations with {\num{2000} slow $\times$ \num{100000}} fast-fading samples.
For comparison, the dashed line shows the approximation for large $N$. %
\begin{figure}
	\centering
	{\begin{tikzpicture}
	\begin{axis}[
		width=.97\linewidth,
		height=.25\textheight,
		xlabel={Transmission Rate $R$},
		ylabel={Outage Probability $\varepsilon_{\text{quant,NLOS}}$},
		ylabel near ticks,
		xmin=1.5,
		xmax=4,
		legend pos = north west,
		legend cell align=left,
		ymajorgrids,
		xmajorgrids,
		xminorgrids,
		grid style={line width=.1pt, draw=gray!20},
		major grid style={line width=.25pt,draw=gray!40},
		]
		\addplot[plot3,mark=*,thick,mark repeat=200] table[x={rate},y={ecdf2}] {data/quant-phase-N20-p0.500.dat};
		\addlegendentry{$K=2$};
		\addplot[plot0,mark=triangle*,thick,mark repeat=200, mark phase=50] table[x={rate},y={ecdf3}] {data/quant-phase-N20-p0.500.dat};
		\addlegendentry{$K=3$};
		\addplot[plot2,mark=square*,thick,mark repeat=200, mark phase=100] table[x={rate},y={ecdf10}] {data/quant-phase-N20-p0.500.dat};
		\addlegendentry{$K=10$};
		\addplot[plot2!70!black,mark=,very thick, dashed, const plot] table[x={rate},y={approx}] {data/out-prob-random-phase-nlos-N20-p0.500.dat};
		\addlegendentry{Approximation};
	\end{axis}
\end{tikzpicture}}
	\vspace*{-.5em}
	\caption{Outage probability for an \gls{nlos} scenario. The phases of the \gls{ris} with $N{=20}$ elements are quantized with $K$ steps. The phases are randomly and uniformly varied. {The connection probabilities are $p_i=p=0.5$ for all $i=1, \dots{}, N$.} The solid lines show the \gls{ecdf} obtained by Monte Carlo simulations. The dashed line indicates the approximation for large $N$.}
	\label{fig:out-prob-quant-phases-n20}
\end{figure}
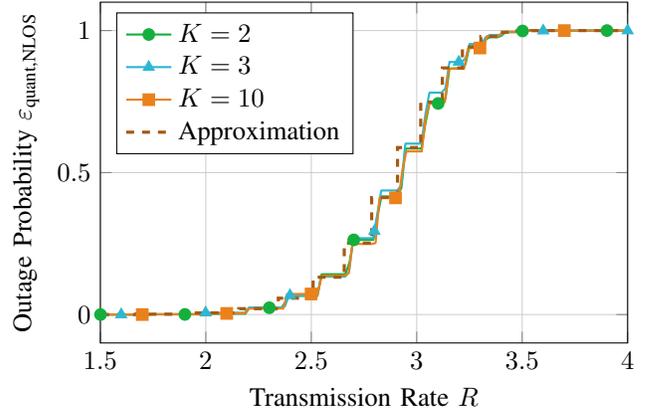
A different perspective of the results is given in Fig.~\ref{fig:out-prob-quant-phases-n20-k2}, where we fix the number of quantization levels {to $K=2$ and vary the connection probability $p$}.
It can be seen that the outage probability curves get closer to a step function with increasing~$p$.
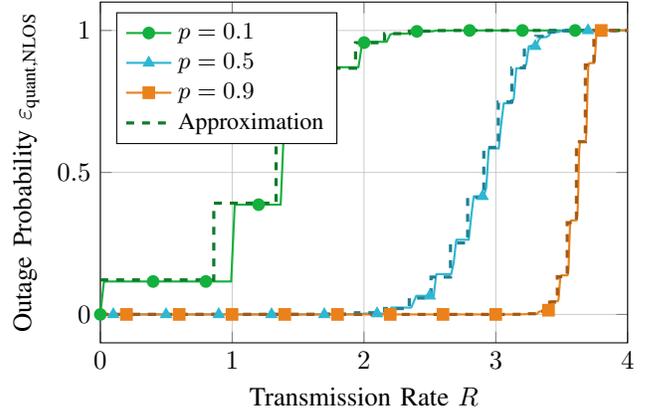
\begin{figure}
	\centering
	{\begin{tikzpicture}
	\begin{axis}[
		width=.97\linewidth,
		height=.25\textheight,
		xlabel={Transmission Rate $R$},
		ylabel={Outage Probability $\varepsilon_{\text{quant,NLOS}}$},
		ylabel near ticks,
		xmin=0,
		xmax=4,
		legend pos = north west,
		legend cell align=left,
		legend style={font=\small},
		ymajorgrids,
		xmajorgrids,
		xminorgrids,
		grid style={line width=.1pt, draw=gray!20},
		major grid style={line width=.25pt,draw=gray!40},
		]
		\addplot[plot3,mark=*,thick,mark repeat=200] table[x={rate},y={ecdf2}] {data/quant-phase-N20-p0.100.dat};
		\addlegendentry{$p=0.1$};
		\addplot[plot0,mark=triangle*,thick,mark repeat=200, mark phase=50] table[x={rate},y={ecdf2}] {data/quant-phase-N20-p0.500.dat};
		\addlegendentry{$p=0.5$};
		\addplot[plot2,mark=square*,thick,mark repeat=200, mark phase=100] table[x={rate},y={ecdf2}] {data/quant-phase-N20-p0.900.dat};
		\addlegendentry{$p=0.9$};
		
		\addplot[plot3!70!black,mark=,very thick, dashed, const plot] table[x={rate},y={approx}] {data/out-prob-random-phase-nlos-N20-p0.100.dat};
		\addlegendentry{Approximation};
		\addplot[plot0!70!black,mark=,very thick, dashed, const plot] table[x={rate},y={approx}] {data/out-prob-random-phase-nlos-N20-p0.500.dat};
		\addplot[plot2!70!black,mark=,very thick, dashed, const plot] table[x={rate},y={approx}] {data/out-prob-random-phase-nlos-N20-p0.900.dat};
	\end{axis}
\end{tikzpicture}}
	\vspace*{-.5em}
	\caption{Outage probability for an \gls{nlos} scenario. The phases of the \gls{ris} with $N={20}$ elements are quantized with $K{=2}$ steps and randomly varied. {The connection probability of all links is $p$.} The solid lines show the \gls{ecdf} obtained by Monte Carlo simulations. The dashed curves indicate the theoretical value without quantization {for large $N$}.}
	\label{fig:out-prob-quant-phases-n20-k2}
\end{figure}

\begin{rem}
	It should be emphasized that the quantization does not change the ergodic capacity for sufficiently large $N$.
	This includes the one-bit-quantization $K=2$, i.e., $\Q=\left\{0, \pi\right\}$, which can already be implemented with existing hardware~\cite{Kaina2014}.
	Furthermore, this approximation is valid since we typically have \gls{ris} with $N>50$ elements.
\end{rem}
\section{{Comparison to Different Phase Adjustment Schemes}}\label{sec:comparison}
{%
After evaluating the outage performance for \gls{ris} phase hopping, we now want to compare it with other methods to adjust the phases of the \gls{ris} elements.
The first scheme is fixing the phases to constant values for the whole transmission.
This is a natural idea since we do not assume perfect \gls{csi} at the transmitter and \gls{ris} and it is therefore not immediately clear how the \gls{ris} phases should be adjusted.

In contrast, we also compare phase hopping to the case of perfectly adjusted \gls{ris} phases.
This is the best case and therefore serves as an upper bound on the performance.
However, it should be emphasized that this requires perfect \gls{csi} at the \gls{ris}, which we do not need for phase hopping.
}

\subsection{Static RIS Phases}\label{sub:comparison-static}
A straightforward way would be to fix each phase to a constant value $\theta_i$.
Without loss of generality, we assume that $\theta_i=0$ for all $i=1, \dots{}, N$.

This yields the effective channel for static \gls{ris} phases as
\begin{equation}\label{eq:channel-static-phases}
	H_{\text{stat}} = \hLOS + \sum_{i=1}^{N}{c_i}\exp\left(\imag \tvp_i\right)\,,
\end{equation}
where we again use the shorthand $\tvp_i = \varphi_i + \psi_i \mod 2\pi$.

{Similar to the results for \gls{ris} phase hopping from Section~\ref{sec:varying-phases}, we can derive the outage probability for static \gls{ris} phases as follows.}

\begin{lem}[{Outage Probability with Static Phases {with $\tilde{N}$ Active Links}}]\label{lem:outage-prob-static-ntilde}
	Consider the previously described \gls{ris}-assisted slow fading communication scenario without \gls{csi} at the transmitter and \gls{ris}.
	The phases of the $N$ \gls{ris} elements $\theta_i$ are fixed to constant values.
	{Let $\tilde{N}$ out of all $N$ links be available, i.e., $\abs{h_i}=\abs{g_i}=1$ for $i=1, \dots{}, \tilde{N}$ with $\tilde{N}\leq N$.}
	For the \gls{nlos} scenario, i.e., $\hLOS=0$, the outage probability is exactly given by
	\begin{equation}\label{eq:outage-prob-static-nlos-exact-ntilde}
		\varepsilon_{\text{stat,NLOS}}({\tilde{N}}) = \sqrt{2^R-1} \int\limits_{0}^{\infty} J_1\left(\sqrt{2^R-1}\cdot t\right) J_0\left(t\right)^{{\tilde{N}}} \diff{t}\,.
	\end{equation}
	For large ${\tilde{N}}$, it can be approximated by
	\begin{equation}\label{eq:outage-prob-static-nlos-approx-ntilde}
		\varepsilon_{\text{stat,NLOS}}({\tilde{N}}) \approx 1 - \exp\left(-\frac{2^R-1}{{\tilde{N}}}\right)\,.
	\end{equation}
	For the \gls{los} scenario, the outage probability can be approximated by
	\begin{equation}\label{eq:outage-prob-static-los-approx-ntilde}
		\varepsilon_{\text{stat,LOS}}{\left(\tilde{N}\right)} \approx 1 - Q_1\left(\sqrt{\frac{2a^2}{{\tilde{N}}}}, \sqrt{\frac{2}{{\tilde{N}}}\left(2^R-1\right)}\right)\,,
	\end{equation}
	where $Q_M(a, b)$ denotes the Marcum Q-function~\cite{Temme1993}.
\end{lem}
\begin{proof}
	{The proof is based on the proof of Lemmas~\ref{lem:erg-cap-varying-nlos-exact-n-tilde} and \ref{lem:erg-cap-varying-nlos-n-tilde} and Theorem~\ref{thm:outage-prob-varying-los} and is therefore omitted.}
\end{proof}

{%
Similarly to Theorem~\ref{thm:outage-prob-varying}, we now need to incorporate the probabilities that only $\tilde{N}$ out of $N$ links are available.
For equal connection probabilities $p_i$ of all links, this yields the following outage probability.
\begin{thm}[{Outage Probability with Static Phases}]\label{thm:outage-prob-static}
	Consider the previously described \gls{ris}-assisted slow fading communication scenario without \gls{csi} at the transmitter and \gls{ris}.
	The phases of the $N$ \gls{ris} elements $\theta_i$ are fixed to constant values.
	{The connection probabilities for all links are the same, i.e., $p_i=p$ for all $i=1, \dots{}, N$.}
	The outage probability is then given by
	\begin{equation}\label{eq:outage-prob-static}
		\varepsilon_{\text{stat}} = \sum_{i=0}^{N} \varepsilon_{\text{stat}}(i) \binom{N}{i}p^{i} (1-p)^{N-i}\,,
	\end{equation}
	where $\varepsilon_{\text{stat}}(i)$ is evaluated according to Lemma~\ref{lem:outage-prob-static-ntilde}.
\end{thm}
\begin{proof}
	The proof follows the idea of the proof of Theorem~\ref{thm:outage-prob-varying} and is therefore omitted.
\end{proof}
}

{%
Similarly to \eqref{eq:erg-capac-varying-nlos-exact-hankel}, the exact outage probabilities for $\tilde{N}$ links from \eqref{eq:outage-prob-static-nlos-exact-ntilde} can be efficiently calculated using the Hankel transform as
\begin{equation*}
	\varepsilon_{\text{stat,NLOS}}\left(\tilde{N}\right) = \sqrt{2^R-1} \hankel_1\left\{\frac{J_0(t)^{\tilde{N}}}{t}\right\}\left(\sqrt{2^R-1}\right)\,.
\end{equation*}
An implementation together with the simulations can be found in \cite{BesserGitlab}.
}

\begin{figure}[t]
	\centering
	{\begin{tikzpicture}
	\begin{axis}[
		width=.98\linewidth,
		height=.27\textheight,
		xlabel={Transmission Rate $R$},
		ylabel={Outage Probability $\varepsilon_{\text{stat,NLOS}}$},
		ylabel near ticks,
		legend pos = south east,
		legend cell align=left,
		legend style={font=\small},
		legend columns = 1,
		legend entries = {{$p=0.1$ -- \Gls{ecdf}}, {$p=0.1$ -- Approx.}, {$p=0.5$ -- \Gls{ecdf}}, {$p=0.5$ -- Approx.}, {$p=0.9$ -- \Gls{ecdf}}, {$p=0.9$ -- Approx.}},
		xmin=0,
		xmax=6,
		mark options={solid},
		ymajorgrids,
		xmajorgrids,
		xminorgrids,
		grid style={line width=.1pt, draw=gray!20},
		major grid style={line width=.25pt,draw=gray!40},
	]
	\addplot[plot3,mark=*,thick, mark repeat=100] table[x={rate},y={ecdf}] {data/out-prob-const-phase-nlos-N20-p0.100.dat};
	\addplot[plot3!70!black,mark=o, very thick, dashed, mark repeat=100] table[x={rate},y={approx}] {data/out-prob-const-phase-nlos-N20-p0.100.dat};

	\addplot[plot0,mark=triangle*,thick, mark repeat=100] table[x={rate},y={ecdf}] {data/out-prob-const-phase-nlos-N20-p0.500.dat};
	\addplot[plot0!70!black,mark=triangle, very thick, dashed, mark repeat=100] table[x={rate},y={approx}] {data/out-prob-const-phase-nlos-N20-p0.500.dat};
	
	\addplot[plot2,mark=square*,thick, mark repeat=100] table[x={rate},y={ecdf}] {data/out-prob-const-phase-nlos-N20-p0.900.dat};
	\addplot[plot2!70!black,mark=square, very thick, dashed, mark repeat=100] table[x={rate},y={approx}] {data/out-prob-const-phase-nlos-N20-p0.900.dat};
	\end{axis}
\end{tikzpicture}}
	\vspace*{-.5em}
	\caption{Outage probability for an \gls{ris} with $N{=20}$ elements with constant phases. {The connection probability of all links is $p$.} The solid lines show the \gls{ecdf} obtained by Monte Carlo simulations. For comparison, the dashed lines indicate the approximation for large $N$.}
	\label{fig:out-prob-const-nlos}
\end{figure}
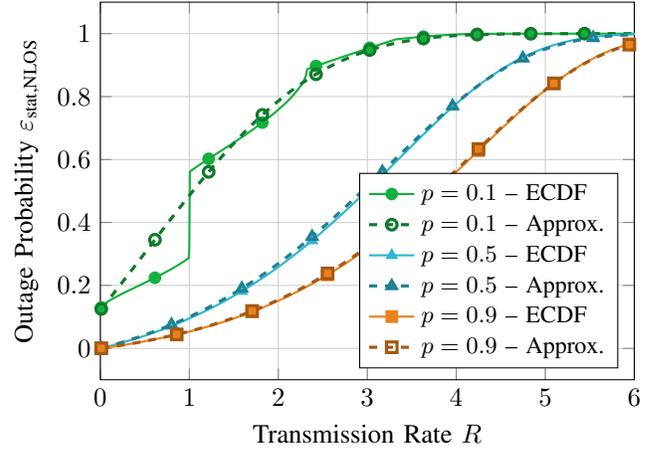
In Fig.~\ref{fig:out-prob-const-nlos}, we show the outage probability $\varepsilon_{\text{stat,NLOS}}$ for an \gls{ris}-assisted communication with static phases and no \gls{los} connection.
The results are shown for a \gls{ris} with $N{=20}$ elements {and different connection probabilities $p$}.
The solid lines are the \glspl{ecdf} obtained by \gls{mc} simulations with $10^6$ samples.
For comparison, the approximation using \eqref{eq:outage-prob-static-nlos-approx-ntilde} is shown.
The source code to reproduce all results can be found at \cite{BesserGitlab}.
From the figure, it can be observed that the approximation is accurate for ${p\geq 0.5}$.
{For small $p$, the approximation is not as good, especially for small rates $R$.
The reason behind this is that it is derived from an approximation for large $\tilde{N}$ based on the central limit theorem.
When $p$ is small, the system is more likely to have a small number of active links $\tilde{N}$, where the approximation is worse.
In this case, the exact value from \eqref{eq:outage-prob-static-nlos-exact-ntilde} should be used.
Additionally, for $p=0.1$, a step-like behavior can be observed.
The reason behind this is similar to the one explained above.
For small $p$, the influence of the outage probabilities for small $\tilde{N}$ is dominant.
The outage probability for $R\to 0$ is determined by the \gls{pdf} of the Binomial distribution for $\tilde{N}=0$.
For $p=0.1$ this is around $0.12$ and therefore clearly visible in Fig.~\ref{fig:out-prob-const-nlos}.
On the contrast, for $p=0.5$, this is only around $10^{-6}$.
The step at $R=1$ is due to the outage probability $\varepsilon_{\text{stat,NLOS}}(1)$ for $\tilde{N}=1$, which is a step function from \num{0} to \num{1} at $R=1$\footnote{{For the evaluation of \eqref{eq:outage-prob-static-nlos-exact-ntilde} for $\tilde{N}=1$, one can use expression of $I^{0}_{01}(1, \sqrt{2^R-1}, 0)$ in \cite[Table~2]{Kausel2012}.}}.
Again, this step is more pronounced for small $p$ due to the Binomial distribution.
}

\subsection{{Perfect Phase Adjustment}}\label{sub:comparison-perfect}
{%
The next \gls{ris} phase adjustment scheme is perfect phase adjustment.
This scheme provides an upper bound on the outage performance, but it should be emphasized that this requires perfect \gls{csi} at the \gls{ris}, which we do \emph{not} assume for phase hopping.
Based on the results from \cite[Lem.~1]{Bjornson2020}, we can directly provide the outage probability for the \gls{nlos} scenario with perfect \gls{ris} phase adjustment in the following.

\begin{cor}[{Outage Probability \Gls{nlos} with Perfect \Gls{ris} Phase Adjustment}]\label{cor:outage-prob-perfect-nlos}
	Consider the previously described \gls{ris}-assisted slow fading communication scenario with perfect \gls{csi} at the \gls{ris}.
	There is no \gls{los} connection, i.e., $\hLOS=0$.
	The phases of the $N$ \gls{ris} elements $\theta_i$ are adjusted to $\theta_i=-\tvp_i$.
	The connection probabilities for all links are the same, i.e., $p_i=p$ for all $i=1, \dots{}, N$.
	The outage probability is then given by
	\begin{align}
		\varepsilon_{\text{perf,NLOS}}
		&= F_{\binomialdist}\left(\sqrt{2^R-1}; N, p\right)\notag\\
		&= \sum_{i=0}^{\floor*{\sqrt{2^R-1}}}\binom{N}{i} p^{i} (1-p)^{N-i}\label{eq:outage-prob-perfect-nlos}
	\end{align}
	where $F_{\binomialdist}(\cdot; N, p)$ denotes the \gls{cdf} of a binomial distribution with $N$ independent trials and success probability $p$.
\end{cor}
}

{%
The three different phase adjustment schemes are now compared in terms of outage probability in Fig.~\ref{fig:out-prob-comparison-n20}.
The scenario is \gls{nlos} with $N=20$ \gls{ris} elements and connection probability $p=0.5$ for all links.
The first observation is that the outage probability curve for static phases is very steep.
This implies that the $\varepsilon$-outage capacity is very small and increases slowly for small $\varepsilon$.
For a tolerated outage probability of $\varepsilon=10^{-5}$, the $\varepsilon$-outage capacity for static \gls{ris} phases is less than \num{0.005}.
On the other hand, it is around \num{0.86} for \gls{ris} phase hopping and \num{1} for perfect phase adjustment.
Even though the outage probability increases faster for phase hopping than perfect phase adjustment, this shows that the proposed phase hopping scheme achieves a similar performance for small $\varepsilon$ as the best-case (perfect phase adjustment).
It should be emphasized that this can be achieved \emph{without} requiring \gls{csi} at the \gls{ris}.
Especially for ultra-reliable communications, we are interested in these very small tolerated outage probabilities, which makes \gls{ris} phase hopping a viable candidate to achieve the required performance.

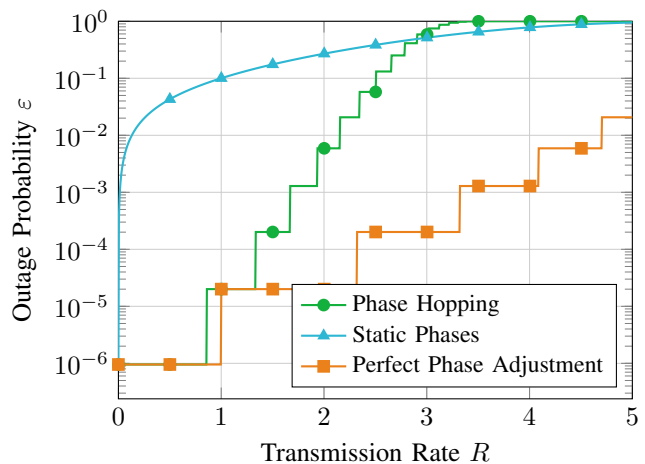
\begin{figure}[t]
	\centering
	{\begin{tikzpicture}
	\begin{axis}[
		width=.95\linewidth,
		height=.27\textheight,
		xlabel={Transmission Rate $R$},
		ylabel={Outage Probability $\varepsilon$},
		ylabel near ticks,
		legend pos = south east,
		legend cell align=left,
		legend columns = 1,
		legend style={font=\small},
		xmin=0,
		xmax=5,
		ymode=log,
		ymax=1,
		mark options={solid},
		ymajorgrids,
		xmajorgrids,
		xminorgrids,
		grid style={line width=.1pt, draw=gray!20},
		major grid style={line width=.25pt,draw=gray!40},
		]
		\addplot[plot3,mark=*,thick, mark repeat=100] table[x={rate},y={hopping}] {data/out-prob-comparison-N20-p0.500.dat};
		\addlegendentry{Phase Hopping};
		\addplot[plot0,mark=triangle*,thick, mark repeat=100] table[x={rate},y={static}] {data/out-prob-comparison-N20-p0.500.dat};
		\addlegendentry{Static Phases};
		\addplot[plot2,mark=square*,thick, mark repeat=100] table[x={rate},y={perfect}] {data/out-prob-comparison-N20-p0.500.dat};
		\addlegendentry{Perfect Phase Adjustment};
	\end{axis}
\end{tikzpicture}}
	\vspace*{-.75em}
	\caption{{Outage probability for a \gls{ris}-assisted \gls{nlos} scenario with $N=20$ \gls{ris} elements with different phase adjustment schemes. The connection probabilities are $p_i=p=0.5$ for all $i=1, \dots{}, N$.}}
	\label{fig:out-prob-comparison-n20}
\end{figure}
}
\section{Conclusion}\label{sec:conclusion}
\begin{table*}%
	\centering
	\caption{Summary of the Outage Performance Results for {Different \Gls{ris} Phase Adjustment Schemes}.}
	\label{tab:summary-outage-prob}
		{%
			\begin{tabularx}{\linewidth}{m{.15\linewidth}|MM|M}
				\toprule
				\multirow[c]{2}{*}{\textbf{Scheme}} & \multicolumn{2}{c|}{\textbf{Non-Line-of-Sight}} & {\textbf{Line-of-Sight}}\\
				& \textbf{Exact} & \textbf{Approximation} & \textbf{Approximation}\\
				\midrule
				{Phase Hopping} &
				\multicolumn{2}{c|}{Theorem~\ref{thm:outage-prob-varying}} &%
			\multirow[c]{2}{*}{Theorem~\ref{thm:outage-prob-varying-los}}\\
			(Section~\ref{sec:varying-phases} and \ref{sec:quant-phases}) &
			Lemma~\ref{lem:erg-cap-varying-nlos-exact-n-tilde}: Equation~\eqref{eq:erg-cap-varying-nlos-exact} & Lemma~\ref{lem:erg-cap-varying-nlos-n-tilde}: Equation~\eqref{eq:erg-cap-varying-nlos-approx} &
			\\[1.2em]
			Static Phases &
			\multicolumn{2}{c|}{Theorem~\ref{thm:outage-prob-static}} &%
		{Theorem~\ref{thm:outage-prob-static}}%
	\\
	(Section~\ref{sub:comparison-static}) &
	Lemma~\ref{lem:outage-prob-static-ntilde}: Equation~\eqref{eq:outage-prob-static-nlos-exact-ntilde} &
	Lemma~\ref{lem:outage-prob-static-ntilde}: Equation~\eqref{eq:outage-prob-static-nlos-approx-ntilde} &
	Lemma~\ref{lem:outage-prob-static-ntilde}: Equation~\eqref{eq:outage-prob-static-los-approx-ntilde}
	\\
	\bottomrule
\end{tabularx}
}
\end{table*}

In this work, we investigate a phase hopping scheme for \gls{ris}-assisted communication systems, which converts slow fading into artificial fast fading.
We showed how this helps to significantly improve the outage performance, {e.g., in terms of the $\varepsilon$-outage capacity}.
{In particular, for small tolerated outage probabilities, i.e., in the context of ultra-reliable communications, this scheme performs close to the best-case.}
Other advantages of the proposed scheme are that it neither requires \gls{csi} at the transmitter and the \gls{ris} nor any additional communication overhead to the \gls{ris}.

A summary of the results for the outage probabilities in the various considered scenarios can be found in Table~\ref{tab:summary-outage-prob}.

{%
For a practical implementation, it needs to be verified that the technology allows for such rapid phase adjustments.
Additionally, the phase adjustments are likely to require additional power which decreases the overall energy efficiency.
This factor needs to be incorporated in future analysis.
However, in contrast to setting specific phases at the \gls{ris}, the proposed scheme applies random phases.
Therefore, the configuring of the phase shifts does not need to be very accurate.
The presented work could also be extended by an additional beamforming optimization for a multi-antenna system in future work.
}

\appendices
{%
\section{Proof of Theorem~\ref{thm:out-prob-general-fading}}
\label{app:general-distributions}
	By the central limit theorem, the overall channel gain $H$ can be approximated for large $N$ as~\cite[Chap.~3.4]{Jammalamadaka2001}
	\begin{equation*}
		H = \sum_{i=1}^{N}\abs{h_i}\abs{g_i} \exp\left(\imag \left(\theta_i + \tvp_i\right)\right) \approx \bar{C} + \imag \bar{I}
	\end{equation*}
	with independent $\bar{C}$ and $\bar{I}$, which are normally distributed according to $\bar{C}, \bar{I}\sim\normaldist(0, \sigma^2)$.
	The variance is given as
	\begin{equation*}
		\sigma^2 = \frac{1}{2}\sum_{i=1}^{N}\abs{h_i}^2\abs{g_i}^2\,.
	\end{equation*}
	The absolute value $\abs{H}$ is, therefore, approximated by a Rayleigh distribution, i.e., $\abs{H}\sim\rayleigh(\sigma)$.
	The ergodic capacity from \eqref{eq:def-erg-capac} for this approximation is then
	\begin{align*}
		C_{\text{erg}}
		&= \expect[\abs{H}]{\log_2\left(1 + \abs*{H}^2\right)}\\
		&= \int_{0}^{\infty} \log_2\left(1+s^2\right) f_{s}(s)\diff{s}\\
		&= \frac{-1}{\log(2)}\exp\left(\frac{1}{2\sigma^2}\right)\Ei\left(-\frac{1}{2\sigma^2}\right)\,.
	\end{align*}
	From the definition of the outage probability in \eqref{eq:def-outage-prob-varying}, it follows that
	\begin{align*}
		\varepsilon_{\text{NLOS}}
		&= \Pr\left(\frac{-1}{\log(2)}\exp\left(\frac{1}{2\sigma^2}\right)\Ei\left(-\frac{1}{2\sigma^2}\right) < R\right)\\
		&= \Pr\left(\frac{1}{2\sigma^2} > \inv{\E}\left(R \log(2)\right)\right)\\
		&= F_{\sigma^2}\left(\frac{1}{2\inv{\E}\left(R \log(2)\right)}\right)\,,
	\end{align*}
	where we introduce the function $\E:\reals_+\to\reals_+$ with $\E(x)=-\exp(x)\Ei(-x)$\footnote{{The function $-\Ei(-x)$ is also known as $E_1(x)$~\cite[Chap.~5]{Abramowitz1972}. From this, it is straightforward to show that $\E(x)$ is strictly monotonic decreasing by verifying that its derivative is negative. Hence, $\E$ has a unique inverse $\inv{\E}$.}}.
}
\section{Proof of Theorem~\ref{thm:outage-prob-varying-los}}\label{app:proof-thm-outage-prob-los}

Similarly to {Lemma~\ref{lem:erg-cap-varying-nlos-n-tilde}}, we will use an approximation of the channel coefficient $H$ for large ${\tilde{N}}$ in the following.
{Therefore, we will first assume a fixed number of active connections $\tilde{N}$.}
{For uniformly distributed $\theta_i$ and $\tvp_i$}, we obtain the approximation by the central limit theorem that for large ${\tilde{N}}$~\cite[Chap.~3.4]{Jammalamadaka2001}
\begin{equation}\label{eq:approx-sum-exp-static}
	\sum_{i=1}^{{\tilde{N}}}\exp\left(\imag\left(\tvp_i{+\theta_i}\right)\right) = {\tilde{N}} \bar{C} + \imag {\tilde{N}} \bar{I}
\end{equation}
with independent $\bar{C}$ and $\bar{I}$, which are normally distributed according to $\bar{C}, \bar{I}\sim\normaldist(0, \frac{1}{2{\tilde{N}}})$.
This yields the following approximation for $H$
\begin{align}
	H &= \left(a\cos\pLOS + {\tilde{N}}\bar{C}\right) + \imag\left(a\sin\pLOS + {\tilde{N}}\bar{I}\right)\\
	&= {\tilde{N}}\left(\hat{C} + \imag\hat{I}\right)\,,
\end{align}
with
\begin{equation*}
	\hat{C}\sim\normaldist\left(\frac{a\cos\varphi_{\text{LOS}}}{{\tilde{N}}}, \frac{1}{2{\tilde{N}}}\right)
	\!\enspace
	\text{and}
	\enspace
	\hat{I}\sim\normaldist\left(\frac{a\sin\varphi_{\text{LOS}}}{{\tilde{N}}}, \frac{1}{2{\tilde{N}}}\right).
\end{equation*}
With this, we can derive that
\begin{equation}
	Z = \frac{2 \abs{H}^2}{{\tilde{N}}} = 2{\tilde{N}}\left(\hat{C}^2 + \hat{I}^2\right)
\end{equation}
is distributed according to a noncentral $\chi^2$ distribution with \num{2} degrees of freedom and noncentrality parameter $\frac{2a^2}{{\tilde{N}}}$, i.e., $Z\sim\chi^2\left(2, \frac{2a^2}{{\tilde{N}}}\right)$~\cite[Chap.~12]{Forbes2010}.

From this, we can derive the \gls{pdf} of $\abs{H}^2$ as
\begin{align}
	f_{\abs{H}^2}(s)
	&= \frac{1}{{\tilde{N}}}\exp\left(\frac{-(a^2+s)}{{\tilde{N}}}\right) I_0\left(\frac{2a}{{\tilde{N}}}\sqrt{s}\right)\label{eq:pdf-gain-varying-los}\,,
\end{align}
where $I_0$ denotes the modified Bessel function of the first kind and order zero~\cite[Eq.~9.6.16]{Abramowitz1972}.
The ergodic capacity is then calculated as
\begin{align*}
	&C_{\text{erg,LOS}}{\left(\tilde{N}\right)}
	= \expect{\log_2\left(1 + \abs{H}^2\right)}\\
	&\quad= \int_{0}^{\infty}\frac{1}{{\tilde{N}}}\log_2(1+s)\exp\left(\frac{-(a^2+s)}{{\tilde{N}}}\right) I_0\left(\frac{2a}{{\tilde{N}}}\sqrt{s}\right)\diff{s}%
	\,.
\end{align*}
It is clear to see that this ergodic capacity does not depend on the realizations of $\hLOS$, $\vec{h}$, and $\vec{g}$.
Thus, the outage probability according to \eqref{eq:def-outage-prob-varying} is a step function
\begin{equation}\label{eq:proof-outage-prob-los}
	\varepsilon_{\text{LOS}}{\left(\tilde{N}\right)} =
	\step\left(R - C_{\text{erg,LOS}}{\left(\tilde{N}\right)}\right)\,.
\end{equation}
{%
Note that this is the outage probability for a fixed number of active link $\tilde{N}$.
In the next step, we need to take the probability into account that $\tilde{N}$ out of $N$ links are active.

Hence, the overall outage probability is given as
\begin{align*}
	\varepsilon_{\text{var,LOS}}
	&= \sum_{i=0}^{N} \varepsilon_{\text{LOS}}\left(i\right) \Pr\left(\tilde{N}=i\right)\\
	&= \sum_{i=0}^{N} \step\left(R - C_{\text{erg,LOS}}\left(i\right)\right) \binom{N}{i}p^{i}(1-p)^{N-i}
\end{align*}
where the last line follows from \eqref{eq:proof-outage-prob-los} and the fact that $\tilde{N}$ is distributed according to a Binomial distribution with $N$ independent trials of probability $p$.

Finally, we analyze the special case of $\tilde{N}=0$.
In this case, we simply have $H=\hLOS=a\exp(\imag\varphi_{\text{LOS}})$ and, thus, $\abs{H}^2=a^2$.
The ergodic capacity is then $C_{\text{erg,LOS}}(0)=\log_2(1+a^2)$, which concludes the proof.
}
\section{{Proof of Theorem~\ref{thm:outage-prob-quant}}}
\label{app:proof-quant-phases}
First, we will only take a closer look at the \gls{nlos} part of $H$, i.e.,
\begin{equation*}
	\sum_{i=1}^{N} \exp\left(\imag\left(\tilde{\varphi_i} + \theta_i\right)\right)\,.
\end{equation*}
We again assume that the phases $\theta_i$ are \gls{iid} and uniformly distributed over $\Q$ and randomly changing with each transmitted symbol.
The phases $\tvp_i$ on the other hand are (continuously) uniformly distributed over $[0, 2\pi]$.
The above expression can be equivalently expressed as
\begin{equation}
	\sum_{i=1}^{N} \exp\left(\imag\left(\tilde{\varphi_i} + \theta_i\right)\right) = \sum_{i=1}^{N}\cos\left(\tilde{\varphi_i} + \theta_i\right) + \imag \sin\left(\tilde{\varphi_i} + \theta_i\right)\,.
\end{equation}
Due to the uniform quantization of the phases in $\Q$, we obtain
\begin{equation}
	\expect[{\theta_i\sim\unif(\Q)}]{\cos\left(\tilde{\varphi_i} + \theta_i\right)} = \expect[{\theta_i\sim\unif(\Q)}]{\sin\left(\tilde{\varphi_i} + \theta_i\right)} = 0
\end{equation}
and
\begin{equation}
	\var\left[{\cos\left(\tilde{\varphi_i} + \theta_i\right)}\right] = \var\left[{\sin\left(\tilde{\varphi_i} + \theta_i\right)}\right] = 0.5\,.
\end{equation}
Note that it is crucial for the above to have the uniform quantization of the unit circle that we assumed for $\Q$ in \eqref{eq:quant-phases-set}.

Due to the independence, we can apply the central limit theorem for the sums, which yields for large $N$
\begin{equation}\label{eq:sum-cos-norm-approx}
	\sum_{i=1}^{N} \cos\left(\tilde{\varphi_i} + \theta_i\right) \stackrel{N\to\infty}{\sim} \normaldist\left(0, \frac{N}{2}\right)\,.
\end{equation}
The same holds for the sum $\sum_{i=1}^{N} \sin\left(\tilde{\varphi_i} + \theta_i\right)$.
Note that this can be applied for any symmetric quantization of the phases, regardless of $K$.

A numerical validation of this observation is shown in Fig.~\ref{fig:hist-sum-cos}.
Histograms of $\sum_{i=1}^{N} \cos\left(\tilde{\varphi_i} + \theta_i\right)$ are presented for $N=4$ in Fig.~\ref{fig:hist-sum-cos-n-small} and $N=50$ in Fig.~\ref{fig:hist-sum-cos-n-large}.
The number of phase quantization levels is set to $K=4$.
The values are obtained by \gls{mc} simulations with $10^6$ samples~\cite{BesserGitlab}.
For comparison, the \gls{pdf} of the approximate normal distribution from \eqref{eq:sum-cos-norm-approx} is shown.
It is clear to see that the approximation is accurate for $N=50$.
On the other hand, $N=4$ is too small to use the normal distribution from \eqref{eq:sum-cos-norm-approx} as an accurate approximation.
\begin{figure}
	\centering
	\subfloat[{$N=4$}]{\begin{tikzpicture}%
	\begin{axis}[
		width=.47\linewidth,
		height=.23\textheight,
		xlabel={Sum of Cosine Terms},
		ylabel={PDF},
		area style,
		scaled y ticks=false,
		ylabel near ticks,
		xmin=-4,
		xmax=4,
		ymin=0,
		]
		\addplot[plot0,fill=plot0!70!white, ybar interval,mark=no] table[x=bins,y=hist] {data/hist-sum-cos-N4-K4.dat};
		\addplot[plot2, mark=no, very thick] table[x=bins,y=norm] {data/hist-sum-cos-N4-K4.dat};
	\end{axis}
\end{tikzpicture}\label{fig:hist-sum-cos-n-small}}
	\hfill
	\subfloat[{$N=50$}]{\begin{tikzpicture}%
	\begin{axis}[
		width=.47\linewidth,
		height=.23\textheight,
		xlabel={Sum of Cosine Terms},
		ylabel={PDF},
		area style,
		scaled y ticks=false,
		ylabel near ticks,
		xmin=-20,
		xmax=20,
		ymin=0,
		yticklabel style={
			/pgf/number format/precision=3,
			/pgf/number format/fixed
		},
		]
		\addplot[plot0,fill=plot0!70!white,ybar interval,mark=no] table[x=bins,y=hist] {data/hist-sum-cos-N50-K4.dat};
		\addplot[plot2, mark=no, very thick] table[x=bins,y=norm] {data/hist-sum-cos-N50-K4.dat};
	\end{axis}
\end{tikzpicture}\label{fig:hist-sum-cos-n-large}}
	\caption{Histogram of $\sum_{i=1}^{N} \cos\left(\tilde{\varphi_i} + \theta_i\right)$ for two values of $N$ and $K=4$. The solid line indicates the \gls{pdf} of the Gaussian approximation from \eqref{eq:sum-cos-norm-approx}.}
	\label{fig:hist-sum-cos}
\end{figure}
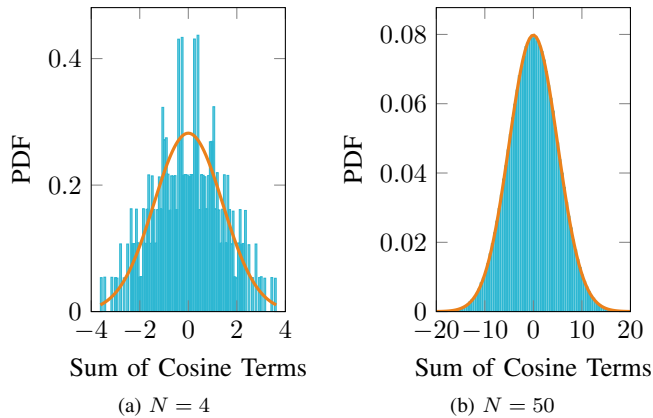

Now that we have established that both real and imaginary part of
\begin{equation*}
	\sum_{i=1}^{N} \exp\left(\imag\left(\tilde{\varphi_i} + \theta_i\right)\right)
\end{equation*}
tend to normal distributions for large $N$, we can apply the results from {Appendix~\ref{app:proof-thm-outage-prob-los}} for the \gls{los} case.
This means that
\begin{equation*}
	\abs{H}^2 = \abs*{\hLOS + \sum_{i=1}^{N} \exp\left(\imag\left(\tilde{\varphi_i} + \theta_i\right)\right)}^2
\end{equation*}
is approximately distributed according to {the \gls{pdf} from \eqref{eq:pdf-gain-varying-los}} for sufficiently large $N$. %
The ergodic capacity is, therefore, also equal to the one of the non-quantized case from~\eqref{eq:erg-capac-random-los-integral}.
It directly follows that the same holds for the \gls{nlos} case.

\printbibliography

\end{document}